\documentclass{article}

\usepackage[T1]{fontenc}
\usepackage[portrait,margin=1in]{geometry}
\usepackage{amssymb,amsthm,amsmath,amssymb,amsfonts,latexsym,graphicx,mathtools,mathrsfs,bbm}
\usepackage{dsfont}
\usepackage{enumerate}

\usepackage[usenames,dvipsnames]{xcolor}
\usepackage{graphics}
\usepackage{tikz, pgfplots, float}
\usepackage{svg}
\pgfplotsset{compat=1.18}

\theoremstyle{definition}

\numberwithin{equation}{section}

\newtheorem{thm}{Theorem}[section]

\newtheorem{corollary}[thm]{Corollary}
\newtheorem{definition}[thm]{Definition}

\newtheorem{lemma}[thm]{Lemma}

\newtheorem{remark}[thm]{\emph{Remark}}

\renewcommand{\Im}{\operatorname*{Im}}

\newcommand{\dee}{\ensuremath{\textrm{d}}}

\newcommand{\inty}[4]{\ensuremath{ \int_{#1}^{#2} \! #3 \, \dee#4 }}

\newcommand{\floor}[1]{ \left\lfloor #1 \right\rfloor }

\newcommand{\braket}[1]{\left\langle #1 \right\rangle}

\usepackage{cite}
\usepackage{hyperref,chngcntr}
\hypersetup{
    colorlinks=true,
    linkcolor=black,
    filecolor=blue,  
    citecolor = black,    
    urlcolor=blue,
    pdftitle={Overleaf Example},
    pdfpagemode=FullScreen,
}

\title{Ergodic properties of incommensurate twisted bilayer materials}

\author{Nathan J. Essner, Jeremiah Williams, and Alexander B. Watson}

\begin{document}

\newpage

\maketitle

\begin{abstract}
    We consider one-dimensional deterministic and random tight-binding Hamiltonians modeling electronic properties of twisted bilayer materials. When the twisted structure is incommensurate, we prove convergence of the density of states measure in the thermodynamic limit and Pastur's theorem on shift-invariance of the spectrum. 
    Our results extend those of Massatt et al. \cite{Massatt2017} and Canc\`es et al. \cite{Cances2017a} in allowing for randomness. 
    We provide numerical density of states computations for the operators we consider.
\end{abstract}



\section{Introduction} \label{sec:intro_0}

\subsection{Motivation}

The field of two-dimensional materials began in 2005 with the isolation and characterization of graphene, a single layer of carbon atoms \cite{2005NovoselovGeimMorozovJiangKatsnelsonGrigorievaDubonosFirsov}. Since then, many other two-dimensional materials have been isolated, such as hexagonal Boron Nitride (hBN), and the transition metal dichalcogenides (TMDCs) \cite{geim_grigorieva}. 

In recent years, attention has shifted to stackings of such materials with a relative twist. For general twist angles, such materials may be aperiodic at the atomic scale because of incommensurability of the layer Bravais lattices. However, for relatively small twist angles, many properties of such materials can be captured by effective continuum PDE models which are periodic over the lattice of interlayer disregistry oscillation known as the moir\'e pattern \cite{Bistritzer2011,CancesGarrigueGontier,Watson2022}. For this reason, such materials are known as moir\'e materials. 

Moir\'e materials have attracted huge attention since the observation of correlated insulating and superconducting electronic phases in twisted bilayer graphene twisted to the ``magic angle'' $\approx 1^\circ$ \cite{Cao2018,Cao2018a} in 2018. Recent years have seen many other quantum many-body electronic phases realized in moir\'e materials \cite{Kennes2021}.

\subsection{Results} 

The present work is concerned with fundamental atomic-scale models of moir\'e materials when the individual layer Bravais lattices are incommensurate so that the model has no exact periodic cell. It is well-known that tight-binding models of such materials have an ergodic structure which may be exploited for efficient numerical computation of material properties \cite{Massatt2017,Cances2017a}. The present work builds on these works as follows. 

First, starting from a standard one-dimensional ``coupled chain'' moir\'e material model (e.g. \cite{Cances2017a}), we derive a reduced model involving only one layer where the effect of interlayer hopping is captured by an effective hopping term whose coefficients are periodic with the period of the other layer. We establish convergence of the empirical density of states measure and Pastur's theorem on shift-invariance of the spectrum for this model using standard methods from the theory of ergodic Schr\"odinger operators \cite{AizenmanWarzel2015}. We also discuss how this model can be formally related to the almost Mathieu operator \cite{SIMON1982463,2017Avron,Avila2006}.

We then return to the full coupled chain model and prove the same results for this model by similar methods, although we emphasize that the ergodic structure of this model is more subtle than the standard case considered in \cite{AizenmanWarzel2015}. Although convergence of the empirical density of states measure was already proved in \cite{Massatt2017,Cances2017a} along similar lines, our proof emphasizing the connection with the standard theory may still be of interest. We finally introduce a disordered incommensurate coupled chain model by adding random onsite potential terms. Here we again prove convergence of the empirical density of states and Pastur's theorem directly using some basic results in ergodic theory.

We conclude with numerical computations of the density of states as a function of the interlayer twist parameter. In the deterministic (but incommensurate) case, these computations show a complicated fractal-like structure reminiscent of the Hofstadter butterfly \cite{PhysRevB.14.2239}. In the random case, these structures are somewhat smoothed out. Note that our results do not rely strongly on our models being one-dimensional. We expect, for example, our results to hold essentially verbatim for more realistic two-dimensional models of moir\'e materials such as twisted bilayer graphene at the cost of more complicated definitions and proofs.

\subsection{Structure of work}

We start by motivating and precisely defining the three models we consider in the present work in Section \ref{sec:TB_models}. We then present our main analytical results on the ergodic properties of these models in Section \ref{sec:TB_ergodic}, postponing the proof of ergodicity of the random coupled chain model to Appendix \ref{sec:proofs_random_iccm}. We present our numerical results in Section \ref{sec:numerics}. For the reader's convenience we discuss details of the numerical method (kernel polynomial method) in Appendix \ref{sec:KPM}.

\subsection{Related work}

As stated above, the ergodic structure of incommensurate bilayer materials was considered in \cite{Massatt2017,Cances2017a}; see also, for example, \cite{Carr2017,MassattCarrLuskinOrtner2018,2018CarrMassattTorrisiCazeauxLuskinKaxiras,Massatt2020,Cazeaux2020,Massatt2021,Wang2025,Quan_2025}, although these works do not consider the reduced chain model we consider here, Pastur's theorem, or the random case. Localization has been proved for a tight-binding model of bilayer graphene where one layer is strained with respect to the other, creating an effectively one-dimensional moir\'e pattern; see \cite{Timmel2020,Becker2025b}. It would be interesting to see whether these results can be extended to any of the models considered here. 

Other recent relevant mathematical works have rigorously established stability of the semimetallic phase of incommensurate twisted bilayer graphene via a renormalization group analysis \cite{jauslin2025incommensuratetwistedbilayergraphene}, studied the Dirac cones of twisted bilayer graphene at commensurate angles \cite{malinovitch2024twistedbilayergraphenecommensurate}, and proposed a higher-order KPM-like method for computing the density of states \cite{yi2025highorderregularizeddeltachebyshevmethod}. It was shown in \cite{2023CancesMeng} that continuum moir\'e-scale models can be realized as semiclassical expansions of atomic-scale models. Numerical comparison between such models with the model of \cite{Wang2025} was carried out in \cite{cances2025numericalcomputationdensitystates}.

\subsection*{Acknowledgements}

This work began as an independent study project at UMN by NJE during Summer and Fall 2023 supervised by ABW. ABW's research was supported in part by NSF grant DMS 2406981. ABW would like to thank Daniel Massatt, Paul Cazeaux, \'Eric Canc\`es, Mitchell Luskin, and Simon Becker for stimulating conversations.

\section{Incommensurate bilayer materials: 1D tight-binding models} \label{sec:TB_models}

\subsection{Model of incommensurate coupled chains} 

In this section, we recall the incommensurate coupled chain model, a 1D model which shares many of the features of tight-binding models of twisted bilayer graphene and other moir\'e materials.
The model describes an electron hopping along two 1D chains, one with lattice constant $1$, and another with lattice constant $1 - \theta$, where $0<\theta<1$ is irrational. Irrationality of $\theta$ implies that the system has no exact periodic cell, so we call the chains incommensurate.

The coupled chain Hilbert space $\ell^2(\mathbb{Z}) \oplus \ell^2(\mathbb{Z})$ consists of infinite vectors defined on the lattices
\begin{equation}
    \phi = \begin{pmatrix} \phi_1\\ \phi_2\end{pmatrix}, \qquad \phi_i = (\phi_{in})_{n\in \mathbb{Z}}, \qquad \phi_{in} \in \mathbb{C}, \label{eq:simple_1d_chain_hilbert_space}
\end{equation}
with the natural inner product
\begin{equation}
    \langle\phi\mid \psi\rangle = \sum_{n\in \mathbb{Z}} \overline{\phi_{1n}}\psi_{1n} + \overline{\phi_{2n}}\psi_{2n}.
    \label{eq:simple_1d_chain_inner_product}
\end{equation}
We define the coupled chain Hamiltonian as follows.

\begin{definition} \label{def:incommensurate_coupled_chain_hamiltonian}
    For each $b \in \mathbb{R}$, known as the \emph{interlayer shift}, let $H(b)$ be the self-adjoint operator
    \begin{equation}
        H(b)\psi = \begin{pmatrix}
        (H(b)\psi)_1\\ (H(b)\psi)_2 
        \end{pmatrix}, \qquad (H(b)\psi)_i = ((H(b)\psi)_{in})_{n\in \mathbb{Z}},
        \label{eq:simple_hamiltonian_incommensurate_chain}
    \end{equation}
    where
    \begin{equation} \label{eq:coupled_chain_H}
    \begin{split}
        ((H(b)\psi)_{1n} & = - (\Delta \psi_1)_n + \sum_{n' \in \mathbb{Z}} h(n-(1-\theta)n'-b)\psi_{2n'}\\
        ((H(b)\psi)_{2n} & = - (\Delta \psi_2)_n + \sum_{n' \in \mathbb{Z}} h((1-\theta)n+b-n')\psi_{1n'},
    \end{split}
    \end{equation}
    where $(\Delta \psi)_n := \psi_{n+1} - 2 \psi_n + \psi_{n-1}$ is the \emph{discrete Laplacian}, and $h$ is the \emph{interlayer hopping function}: an even ($h(-\eta) = h(\eta)$), real, smooth, exponentially-decaying function. Note that these assumptions guarantee $H(b)$ is self-adjoint.
\end{definition}
A physically realistic choice for the interlayer hopping function, which arises from the overlap of exponentially decaying Wannier orbitals \cite{kaxiras_joannopoulos_2019}, is
\begin{equation} \label{eq:inter_hop}
    h(\eta) := A e^{- B \sqrt{ \eta^2 + L^2 }},
\end{equation}
where $A, B, L > 0$ are constants. Here, $L > 0$ represents the interchain distance. With this choice, the hopping amplitude between site $n$ on chain $1$, and site $n'$ on chain $2$, depends only on the intersite distance $\sqrt{ (n - (1 - \theta) n' - b)^2 + L^2 }$. However, we do not expect the specific form of $h$ to modify the essential features of the model. We restrict to nearest-neighbor hopping (discrete Laplacian) for the intralayer hopping for simplicity.

For fixed $b$, the model \eqref{eq:coupled_chain_H}-\eqref{eq:inter_hop} describes an electron hopping along a chain of atoms with positions $\left\{ \left(n,\frac{L}{2}\right)^\top \right\}_{n \in \mathbb{Z}}$, which can hop to a second chain of atoms with positions $\left\{ \left( (1 - \theta) n + b,-\frac{L}{2}\right)^\top \right\}_{n \in \mathbb{Z}}$. Note that we fix the origin in the $x$ direction to coincide with the position of an atom in the top layer without loss of generality.

\subsection{Reduced model of incommensurate coupled chains} \label{sec:reduced}

We will also consider a simplified incommensurate coupled chain model which can be derived as follows. Suppose that we apply a uniform onsite potential $-\mu$ to the second layer of the model in Definition \ref{def:incommensurate_coupled_chain_hamiltonian}, so that we replace the second equation of \eqref{eq:coupled_chain_H} by
    \begin{equation} \label{eq:coupled_chain_H}
        (H(b)\psi)_{2n} = - (\Delta \psi_2)_n - \mu \psi_{2n} + \sum_{n' \in \mathbb{Z}} h((1-\theta)n+b-n')\psi_{1n'}.
    \end{equation}
Recalling that the spectrum of $-\Delta$ is $[0,4]$, we have that for $\mu \notin [0,4]$, $-\Delta-\mu$ is invertible. For such $\mu$, we can obtain via Schur complement the following effective tight-binding model acting on the single chain Hilbert space $\ell^2(\mathbb{Z})$, which we refer to as the \emph{reduced incommensurate coupled chain} model.
\begin{definition}
\label{def:reduced_single_coupled_chain_hamiltonian}
For each $b \in \mathbb{R}$, let $H_r(b)$ be the self-adjoint operator 
\begin{equation} \label{eq:reduced_single_chain}
    ( H_r(b)\psi )_n := - (\Delta \psi)_n - \sum_{n' \in \mathbb{Z}} h(n - (1 - \theta)n' - b) \sum_{n'' \in \mathbb{Z}} (- \Delta - \mu)^{-1}_{n' n''} \sum_{n''' \in \mathbb{Z}} h((1 - \theta) n'' + b - n''') \psi_{n'''},
\end{equation}
where $h$ is an even, smooth, exponentially-decaying function.
\end{definition}
The inverse of the discrete Laplacian can be written out explicitly using the Fourier transform. We do not need this formula, only the basic facts that the inverse is translation-invariant $(- \Delta - \mu)^{-1}_{n + m, n'+m} = (- \Delta - \mu)^{-1}_{n,n'}$ for all integers $m$, inversion-symmetric $(-\Delta-\mu)^{-1}_{-n,-n'} = (-\Delta-\mu)^{-1}_{n,n'}$ and that the inverse is local $| (- \Delta - \mu)^{-1}_{n n'} | \leq A e^{- B |n - n'|}$ for constants $A, B > 0$.
We consider Definition \ref{def:reduced_single_coupled_chain_hamiltonian} as a simplified version of Definition \ref{def:incommensurate_coupled_chain_hamiltonian}. Thus,  the discussion below Definition \ref{def:incommensurate_coupled_chain_hamiltonian} applies to the interlayer hopping function here as well.

The reduced coupled chain operator can be written more simply as
\begin{equation}
    ( H_r(b)\psi )_n := - (\Delta \psi)_n + \sum_{n' \in \mathbb{Z}} \tilde{h}(n-b,n'-b) \psi_{n'},
\end{equation}
where the effective hopping function is defined as
\begin{equation}
    \tilde{h}(x,y) := - \sum_{n' \in \mathbb{Z}} \sum_{n'' \in \mathbb{Z}} h(x - (1 - \theta)n') (- \Delta - \mu)^{-1}_{n' n''} h((1 - \theta) n'' - y).
\end{equation}
Using translation and inversion symmetry of $(-\Delta -\mu)^{-1}$, the effective hopping function satisfies the symmetries
\begin{equation}
    \begin{split}
        &\tilde{h}(x+(1-\theta)m,y+(1 - \theta)m) = h(x,y) \quad \forall m \in \mathbb{Z}, 
        \quad \tilde{h}(-x,-y) = \tilde{h}(x,y).
    \end{split}
\end{equation}
Using exponential decay of $(-\Delta - \mu)^{-1}$ and of $h$ we have the decay estimate
\begin{equation}
    \left| \tilde{h}(x,y) \right| \leq A e^{- B |x - y|}
\end{equation}
for constants $A, B > 0$ (not necessarily the same as above). 

In Section \ref{sec:numerics} we will also consider the further simplified model where the inverse Laplacian is replaced by identity
\begin{equation} \label{eq:reduced_single_chain_noLap}
    ( H_{rc}(b)\psi )_n := - (\Delta \psi)_n + \mu^{-1} \sum_{n' \in \mathbb{Z}} h(n - (1 - \theta)n' - b) \sum_{n'' \in \mathbb{Z}} h((1 - \theta) n' + b - n'') \psi_{n''}.
\end{equation}
We can think of this operator as approximating \eqref{eq:reduced_single_chain} in the large $\mu$ limit where only the leading term in a Neumann series expansion of the inverse Laplacian is kept: $(- \Delta - \mu)^{-1} = \mu^{-1} \left(- \mu^{-1} \Delta - 1 \right)^{-1} \approx - \mu^{-1} + O(\mu^{-2})$.

\subsubsection{Connection to almost Mathieu operator}

It is interesting to note that if we truncate $\tilde{h}$ to the onsite term, we obtain
\begin{equation}
    ( H_r(b)\psi )_n \approx - (\Delta \psi)_n + V(n-b) \psi_{n},
\end{equation}
where
\begin{equation}
    V(x) := - \sum_{n' \in \mathbb{Z}} \sum_{n'' \in \mathbb{Z}} h(x - (1 - \theta)n') (- \Delta - \mu)^{-1}_{n' n''} h((1 - \theta) n'' - x)
\end{equation}
satisfies $V(x + (1 - \theta) m) = V(x)$ $\forall m \in \mathbb{Z}$ and $V(-x) = V(x)$.

If we write $V$ as a Fourier series
\begin{equation}
    V(x) = \sum_{n \in \mathbb{Z}} e^{i \frac{2 \pi n}{1 - \theta} x} \hat{V}(n), \quad \hat{V}(n) := \frac{1}{|1 - \theta|} \inty{0}{1 - \theta}{ e^{- i \frac{2 \pi n}{1 - \theta} x} V(x) }{x},
\end{equation}
and then truncate the Fourier series to the first non-zero frequencies, ignoring any constant term (since it makes no difference to the spectral properties), we arrive at a form of the almost Mathieu operator \cite{AizenmanWarzel2015}
\begin{equation} \label{eq:alm_M}
    ( H_r(b)\psi )_n \approx - (\Delta \psi)_n + 2 \lambda \cos\left( 2 \pi \alpha n - \tilde{b} \right) \psi_{n},
\end{equation}
where
\begin{equation}
    \lambda := \hat{V}(1) = \frac{1}{|1 - \theta|} \inty{0}{1 - \theta}{ e^{- i \frac{2 \pi}{1 - \theta} x} V(x) }{x}, \quad \alpha := \frac{1}{1 - \theta}, \quad \tilde{b} := \frac{2 \pi b}{1 - \theta},
\end{equation}
and we have used even-ness of $V$ to conclude that $\hat{V}(-1) = \hat{V}(1)$.

\subsection{Model of incommensurate coupled chains with random onsite potentials}

In this section, we introduce randomness to the incommensurate coupled chain model to simulate disorder effects. Specifically, we will introduce a potential sampled independently and identically distributed (i.i.d.) on each site of the incommensurate coupled chain operator from a given probability distribution on $\mathbb{R}$. We use the same Hilbert space as for the ordinary incommensurate coupled chain operator in Definition \ref{def:incommensurate_coupled_chain_hamiltonian}.

Let $(\mathbb{R},\mathcal{A}_0,\mathbb{P}_0)$ be a probability space, so that $\mathcal{A}_0$ is a $\sigma$-algebra of subsets of $\mathbb{R}$ and $\mathbb{P}_0$ is a measure on $\mathcal{A}_0$ such that $\mathbb{P}_0(\mathbb{R}) = 1$. Recall that the Cartesian product $\mathbb{R}^\mathbb{Z}$ consists of real sequences $(...,x_{-1},x_0,x_1,...)$ where $x_j \in \mathbb{R}$ for all $j \in \mathbb{R}$. We define the product $\sigma$-algebra $\mathcal{A}$ in the usual way as the $\sigma$-algebra generated by the rectangles $\bigtimes_{j \in \mathbb{Z}} A_j$, where $A_j \in \mathcal{A}_0$ for all $j \in \mathbb{Z}$ and $A_j = \mathbb{R}$ for all but finitely-many $j \in \mathbb{Z}$. From \cite{halmos2013measure} (Theorem B of Section 38), we have that there exists a unique probability measure $\mathbb{P}$ on $\mathcal{A}$ known as the product measure such that for any rectangle $E \in \mathcal{A}$, $\mathbb{P}(E)$ is simply the product of the $\mathbb{P}_0$-measures of each of the (finitely-many) sets in the product $E$ not equal to $\mathbb{R}$.

\begin{definition}
    \label{def:random_iccm_operator}
    Consider the Hilbert space of two incommensurate chains of lattice constants $1$ and $1-\theta$, where $0<\theta<1$ is irrational, and the probability space $(\mathbb{R}, \mathcal{A}_0, \mathbb{P}_0)$ with a compactly-supported Borel probability measure $\mathbb{P}_0$. Let $\omega_i = ( \omega_{ij} )_{j\in \mathbb{Z}}$, $i \in \{1,2\}$ be vectors in the probability space $(\mathbb{R}^\mathbb{Z}, \mathcal{A}, \mathbb{P})$ where $\mathbb{P}$ is the product measure over $(\mathbb{R},\mathcal{A}_0, \mathbb{P}_0)$.

    For each $b \in \mathbb{R}$, let $H_R(b,\omega_1,\omega_2)$ be the self-adjoint operator
    \begin{align}
        H_R(b,\omega_1,\omega_2)\psi = \begin{pmatrix}
            (H_R(b,\omega_1,\omega_2)\psi)_1 \\ (H_R(b,\omega_1,\omega_2)\psi)_2
        \end{pmatrix}, \qquad (H_R(b,\omega_1,\omega_2)\psi)_i = ((H_R(b,\omega_1,\omega_2)\psi)_{in})_{n \in \mathbb{Z}},
    \end{align}
    called the \emph{random incommensurate coupled chain operator}, where 
    \begin{equation} \label{eq:rand_coup_chain}
        \begin{split}
            (H_R(b,\omega_1,\omega_2)\psi)_{1n} = - (\Delta \psi_1)_n + \omega_{1n} \psi_{1n} + \sum_{n' \in \mathbb{Z}}  h(n - (1-\theta)n' - b)\psi_{2n'} , \\
            (H_R(b,\omega_1,\omega_2)\psi)_{2n} = - (\Delta \psi_2)_n + \omega_{2n} \psi_{2n} + \sum_{n' \in \mathbb{Z}} h((1-\theta)n + b - n')\psi_{1n'},
        \end{split}
    \end{equation}
    and $h$ is an interlayer hopping function: an even, real, smooth, exponentially-decaying function. 
\end{definition}
We restrict attention to compactly-supported Borel measures for simplicity: in particular, compact support of the probability measure ensures that $H_R$ can be uniformly bounded over $(b,\omega_1,\omega_2)$. Note that the discussion of Section \ref{sec:reduced} applies equally to \eqref{eq:rand_coup_chain} and would allow for deriving an analogous effective model on a single chain with random coefficients. Even more simply, one could obtain such a model by adding an onsite disorder term $\omega_n \psi_n$ to the model \eqref{eq:reduced_single_chain} directly.

\section{Incommensurate bilayer materials: ergodic properties} \label{sec:TB_ergodic}


\subsection{Definition of an ergodic operator}


We start by briefly reviewing the definition of an ergodic operator acting on $\ell^2(\mathbb{Z})$ following \cite{AizenmanWarzel2015}. Recall that matrix elements of real-valued and measurable functions $f$ of self-adjoint operators $H$ acting in a Hilbert space $\mathcal{H}$ can be conveniently defined through the spectral measure (Appendix A of \cite{AizenmanWarzel2015})
\begin{align}
    \langle \varphi, f(H) \psi \rangle \vcentcolon = \int f(\lambda) \, \text{d} \mu_{\varphi,\psi},
\end{align}
where $\mu_{\varphi,\psi}$ is the spectral measure of $H$ associated to the vectors $\varphi, \psi \in \mathcal{H}$, and $\langle , \rangle$ denotes the $\mathcal{H}$-inner product \cite{AizenmanWarzel2015}. 
We then say that a function $\omega \mapsto H(\omega)$ mapping a probability space $(\Omega,\mathcal{A},\mathbb{P})$ to the set of self-adjoint operators on $\mathcal{H}$ is \emph{weakly measurable} if, for all $f\in L^\infty(\mathbb{R})$ and all $\varphi,\psi \in \mathcal{H}$, the functions $\omega \mapsto \langle \varphi, f(H(\omega))\psi\rangle$ are $\mathbb{P}$-measurable \cite{AizenmanWarzel2015}. 
\begin{definition}
    \label{def:random_operator}
    A self-adjoint \textit{random operator} is an operator-valued weakly measurable function defined on a probability space $(\Omega, \mathcal{A}, \mathbb{P})$ which assigns to every $\omega \in \Omega$ a self-adjoint operator $H(\omega)$ acting in some common (separable) Hilbert space $\mathcal{H}$.
\end{definition}
Recall that a transformation $T:X\to X$ on a measure space $(X,\mathcal{A}, \mu)$ is called \emph{measure-preserving} if for all $A\in \mathcal{A}$, $\mu(T^{-1}A) = \mu(A)$. 
\begin{definition}
    \label{def:ergodic_transformations_prob_space}
    The action of a group of measure-preserving transformations $(T_x)_{x\in I}$ on a probability space $(\Omega, \mathcal{A}, \mathbb{P})$ is \textit{ergodic} if all events $A\in \mathcal{A}$ which are invariant under the group, $T^{-1}_xA = A$ for all $x\in I$, are of probability zero or one.
\end{definition}

Recall that a \emph{graph} is a pair $\mathbb{G} = (V,E)$, where $V$ is the set of \emph{vertices}, and $E$ is a set of pairs of vertices, known as the set of \emph{edges}. A \emph{graph automorphism} is a permutation $\sigma$ of the graph's vertices such that, for all vertices $v_1, v_2 \in V$, if $(v_1,v_2) \in E$ then $(\sigma(v_1),\sigma(v_2)) \in E$. A graph $\mathbb{G}$ is called \emph{vertex-transitive} if, for any vertices $v_1, v_2 \in V$, there exists a graph automorphism $f : \mathbb{G} \rightarrow \mathbb{G}$ such that $f(v_1) = v_2$. Below, we will abuse notation to write $\mathbb{G}$ both for the graph and for the set of vertices of $\mathbb{G}$.

\begin{definition}
    \label{def:standard_ergodic_operator_random_0}
    An \textit{ergodic operator} is a random operator $H(\omega)$ such that
    \begin{enumerate}
        \item The operators act on $\mathcal{H} = \ell^2(\mathbb{G})$, where $\mathbb{G}$ is endowed with a group of vertex-transitive graph automorphisms $\mathcal{S} = (S_x)_{x\in I}$.
        \item The group $\mathcal{S}$ can be represented by a group of measure-preserving transformations $(T_x)_{x\in I}$ whose action on a probability space $(\Omega, \mathcal{A}, \mathbb{P})$ is ergodic. Moreover, for every $x\in I$ and $\omega \in \Omega$, 
        \begin{equation}
            H(T_x\omega) = U_{x,\omega} H(\omega) U^\dagger_{x,\omega} \label{eq:unitary_operator_equivlence_0}
        \end{equation}
        with a unitary operator $U_{x,\omega}$ for which all $\psi \in \ell^2(\mathbb{Z})$
        \begin{equation}
            (U_{x,\omega}\psi)(\xi) = \psi(S_x\xi)e^{i\phi_{x,\omega}(\xi)}
            \label{eq:unitary_cond_standard_ergodic_operator}
        \end{equation}
        having some real-valued phase $\phi_{x,\omega}(\xi)$.
    \end{enumerate}
\end{definition}

\begin{remark}
    \label{remark:ergodic_nomenclature_convention}
    The above definition of an ergodic operator is technically a ``standard ergodic operator'' in \cite{AizenmanWarzel2015}, as opposed to an ``ergodic operator'', which is a random operator without the additional structure of each lattice shift corresponding to each transformation of the probability space. We ignore this distinction here. 
\end{remark}

\subsection{Reduced incommensurate coupled chain model as an ergodic operator}

The following theorem asserts that the reduced incommensurate coupled chain model \eqref{eq:reduced_single_chain} is an ergodic operator in the sense of Definition \ref{def:standard_ergodic_operator_random_0}.
\begin{thm}
    \label{thm:ergodicity_simplified_single_chain}
    Consider the graph $\mathbb{Z}$ endowed with the group of graph automorphisms $\mathcal{S} = (S_x)_{x \in \mathbb{Z}}$ where $S_x$ translates the graph by $x$, i.e., $S_x n := n - x$. For each $x \in \mathbb{Z}$, let
    \begin{equation}
        T_x b = b - x \mod (1 - \theta), \quad ( U_x \psi )_n = \psi_{n - x}, \quad n \in \mathbb{Z}.
    \end{equation}
    Then, $(T_x)_{x \in \mathbb{Z}}$ is a group of measure-preserving transformations representing $\mathcal{S}$ whose action on the probability space $([0,1 - \theta),\mathcal{A},\mathbb{P})$, where $\mathcal{A}$ denotes the Borel sets and $\mathbb{P}$ denotes the normalized Lebesgue measure, is ergodic, and $(U_x)_{x \in \mathbb{Z}}$ is a group of unitary operators on $\ell^2(\mathbb{Z})$ such that
    \begin{equation} \label{eq:shift_shift}
        H_r(T_x b) = U_x H_r(b) U_x^\dagger.
    \end{equation}
    Thus, the reduced incommensurate coupled chain model is an ergodic operator in the sense of Definition \ref{def:standard_ergodic_operator_random_0}.
\end{thm}
\begin{proof}
    That $( T_x )_{x \in \mathbb{Z}}$ is a group of measure-preserving transformations acting ergodically for $\theta$ irrational is standard; see e.g. \cite{Nadkarni1998}. The identity \eqref{eq:shift_shift} can be verified by direct calculation.
\end{proof}

With ergodicity established, we can prove convergence of the density of states measure for the reduced incommensurate coupled chain model in the thermodynamic limit. The density of states measure is a key concept when studying electronic properties of materials. 

We start by recalling the \emph{Riesz representation theorem} (Theorem 2.14 of \cite{rudin1987real}). Recall that a linear functional $I$ on $C_c(\mathbb{R})$ (continuous functions on $\mathbb{R}$ with compact support) is called \emph{positive} if $I(f) \geq 0$ whenever $f \geq 0$. A Borel measure is called \emph{outer regular} on a set $E$, and \emph{inner regular} on $E$, respectively, if
\begin{equation}
    \mu(E) = \inf\{ \mu(U) : U \supset E, U \text{ open} \}, \quad \mu(E) = \sup\{ \mu(K) : K \subset E, K \text{ compact} \}.
\end{equation}
A Borel measure is called a \emph{Radon measure} if it is finite on compact sets, outer regular on Borel sets, and inner regular on all open sets. The Riesz representation theorem states that if $I$ is a positive linear functional on $C_c(\mathbb{R})$, there is a unique Radon measure $\mu$ on $\mathbb{R}$ such that
\begin{equation}
    I(f) = \inty{\mathbb{R}}{}{ f }{\mu}, \quad \forall f \in C_c(\mathbb{R}).
\end{equation}
In particular, Radon measures can be defined via their action on test functions, and vice versa.

\begin{definition} \label{def:finite_DOS_0}
    For positive integers $L$, let $\Lambda_L := [-L,L] \cap \mathbb{Z}$, so that $|\Lambda_L| = 2 L + 1$. Then, for each $b \in [0,1-\theta)$,
    define $H_{r,L}(b) = 1_{\Lambda_L} H_r(b) 1_{\Lambda_L}$, where $1_{\Lambda_L}$ is the characteristic function for the set $\Lambda_L$ and $H_r$ is as in \eqref{eq:reduced_single_chain}. We define the \emph{empirical density of states measure} as the unique Radon measure $\text{d}k_{b,L}$ corresponding to the linear functional 
    \begin{equation}
        g \mapsto \frac{1}{|\Lambda_L|} \sum_{x\in \Lambda_L} \langle \delta_x, g(H_{r,L}(b))\delta_x\rangle = \int_\mathbb{R} g\,\text{d}k_{b,L}
        \label{eq:finite_tdos_ergodic_operator}
    \end{equation}
    acting on $g \in C_c(\mathbb{R})$, where, for each $x \in \mathbb{Z}$, $\delta_x$ is the vector whose $x$th entry is $1$ and all others are $0$. 
\end{definition}
Let $E_{j,b,L}$, $j \in \{1,...,2L+1\}$ denote the eigenvalues of $H_{r,L}(b)$, then straightforward application of the spectral theorem for finite-dimensional Hermitian matrices shows that
\begin{equation}
    \text{d}k_{b,L} = \frac{1}{|\Lambda_L|} \sum_{j = 1}^{2 L + 1} \delta_{E_{j,b,L}},
\end{equation}
where $\delta_{E_{j,b,L}}$ denotes the Dirac measure for the point $E_{j,b,L}$.
The thermodynamic limit density of states is defined as follows.
\begin{definition}
    \label{def:density_of_states}
    The \emph{density of states measure} (DOSM) is the measure $\text{d}k$ defined by the linear functional 
    \begin{equation}
        \label{eq:density_of_states_measure}
        \int_\mathbb{R} g\,\text{d}k := \mathbb{E}(\langle \delta_0, g(H(\cdot)) \delta_0\rangle) = \inty{0}{1-\theta}{ \langle \delta_0, g(H(b)) \delta_0 \rangle }{b}, \quad \forall g \in C_c(\mathbb{R}).
    \end{equation}
\end{definition}

Convergence of the density of states measure in the thermodynamic limit now follows from ergodicity and Birkhoff's ergodic theorem \cite{Krengel1985}.
\begin{corollary}
    Let $\Lambda_L$ and $H_{r,L}$ be as in Definition \ref{def:finite_DOS_0}. Then, for all $b \in [0,1-\theta)$ outside of a measure $0$ set,
    \begin{equation} \label{eq:EDOS}
        \lim_{L \rightarrow \infty} \frac{1}{|\Lambda_L|} \sum_{x \in \Lambda_L} \langle \delta_x, g(H_{r,L}(b))\delta_x\rangle = \mathbb{E}( \langle \delta_0 , g(H_r(b)) \delta_0 \rangle ),
    \end{equation}
    where $\mathbb{E}$ denotes expectation with respect to the normalized Lebesgue measure, i.e.,
    \begin{equation} \label{eq:red_this}
        \mathbb{E}( \langle \delta_0 , g(H_r(b)) \delta_0 \rangle ) := \frac{1}{| 1 - \theta |} \inty{0}{1 - \theta}{ \langle \delta_0 , g(H_r(b)) \delta_0 \rangle }{b}.
    \end{equation}
In particular, we have weak convergence of the empirical density of states measures 
    \begin{equation}
        \lim_{L\to \infty}\int f \,\text{d}{k}_{\omega,L} = \int f \,\text{d}k.
    \end{equation}
\end{corollary}
\begin{proof} See the proof of Theorem 3.15 in \cite{AizenmanWarzel2015}. 
\end{proof}

Another standard consequence of ergodicity is the following, known as Pastur's theorem.
\begin{corollary}
    For each $\theta \in (0,1)$ irrational, there is a set $\Sigma \subset \mathbb{R}$ such that 
    \begin{align}
        \mathbb{P}\left(\{b\in [0,1-\theta): \sigma(H_r(b)) = \Sigma \}\right) = 1,
    \end{align} 
    called the \emph{almost-sure spectrum} of the reduced incommensurate coupled chain operators $H_r(b)$.
\end{corollary}
\begin{proof} See the proof of Theorem 3.10 in \cite{AizenmanWarzel2015}.
\end{proof}
Analogous results hold also for the Lebesgue decomposition of the spectrum into its absolutely continuous, singular continuous, and pure point parts.

\subsection{Convergence of the density of states measure for the incommensurate coupled chain model}
\label{sec:convergence_dosm_iccm}

In this section, we will prove convergence of the empirical density of states measure for the full incommensurate coupled chain model. The fact that the model has two chains makes the proof slightly more involved than the case of the reduced coupled chain model. In particular, we will not prove this model is ergodic in the sense of Definition \ref{def:standard_ergodic_operator_random_0} as an intermediate step. However, the proof follows from the same essential ideas.

We start by defining the density of states for the incommensurate coupled chain truncated to the interval $[-L,L]$ for $L > 0$ simultaneously for each chain, subject to Dirichlet boundary conditions.
\begin{definition} \label{def:truncated_dosm_iccm}
    For each fixed $b \in \mathbb{R}$, and each positive $L > 0$, let $\Lambda_L^1 := \mathbb{Z} \cap [-L,L]$ and $\Lambda_L^2 := \mathbb{Z} \cap \left[\frac{-L - b}{1 - \theta},\frac{L - b}{1 - \theta}\right]$. For $L > |b|$, we have
    \begin{equation}
        |\Lambda_L| := |\Lambda_L^1| + |\Lambda_L^2|, \text{ where } |\Lambda_L^1| = 1 + 2 \floor{L}, \quad |\Lambda_L^2| = 1 + \floor{ \frac{L + b}{1 - \theta} } + \floor{ \frac{L - b}{1 - \theta} }, \label{eq:size_of_truncations}
    \end{equation}
    where $\floor{ \cdot }$ denotes the floor function, i.e., the mapping of a positive number $\eta$ to the largest integer less than or equal to $\eta$. Let
    \begin{equation}
            1_{\Lambda_L} \psi := \begin{pmatrix} (1_{\Lambda_L} \psi)_1 \\ (1_{\Lambda_L} \psi)_2 \end{pmatrix}, \quad ( 1_{\Lambda_L} \psi )_i = \left( ( 1_{\Lambda_L} \psi )_{in} \right)_{n \in \mathbb{Z}}, \quad i \in \{1,2\},
    \end{equation}
    where
    \begin{equation}
                ( 1_{\Lambda_L} \psi )_{1n} = \begin{cases} \psi_{1n} & n \in \Lambda_L^1 \\ 0 & \text{else}, \end{cases} \quad ( 1_{\Lambda_L} \psi )_{2n} = \begin{cases} \psi_{2n} & n_2 \in \Lambda_L^2 \\ 0 & \text{else}. \end{cases}
    \end{equation}
    Then, let $H_L(b) := 1_{\Lambda_L} H(b) 1_{\Lambda_L}$. We define the density of states measure for $H_L(b)$ as the unique Radon measure $\text{d}k_{b,L}$ corresponding to the linear functional acting on $g \in C_c(\mathbb{R})$
    \begin{equation}
        g \mapsto \frac{1}{|\Lambda_L|} \left( \sum_{x \in \Lambda_L^1} \langle \delta_{1x} , g(H_L(b)) \delta_{1x}\rangle + \sum_{x \in \Lambda_L^2} \langle \delta_{2x} , g(H_L(b)) \delta_{2x}\rangle \right) = \int_\mathbb{R} g\,\text{d}k_{b,L},
        \label{eq:finite_dosm_truncated_incommensurate_chain_operator}
    \end{equation}
    where $\delta_{ix}$ is the vector with entry $1$ in the entry corresponding to site $x$ on the $i$th layer for each $x \in \mathbb{Z}$ and $i \in \{1,2\}$ and entries $0$ otherwise.
\end{definition}
Again, the measure \eqref{eq:finite_dosm_truncated_incommensurate_chain_operator} is simply a normalized sum of delta masses at the eigenvalues of the truncated operator $H_L(b)$.

We now introduce the ergodic structure which will allow us to prove convergence of the density of states measure \eqref{eq:finite_dosm_truncated_incommensurate_chain_operator} in the limit $L \rightarrow \infty$. We start with the following lemma, which follows immediately from the definitions.
\begin{lemma}
    \label{lemma:shifting_operators_incommensurate_operator_transformations}
    For each $x \in \mathbb{Z}$, define the operator $U^1_x$, which shifts in layer $1$ by $x$ while leaving layer $2$ unchanged, as 
    \begin{equation}
        U^1_{x} \psi = \begin{pmatrix} ( U^1_{x} \psi )_1 \\ ( U^1_{x} \psi )_2 \end{pmatrix}, \quad ( U^1_{x} \psi )_i = \left( ( U^1_{x} \psi )_{in} \right)_{n \in \mathbb{Z}}, \quad i \in \{1,2\},
    \end{equation}
    where
    \begin{equation}
        ( U^1_x \psi )_{1n} = \psi_{1(n-x)}, \quad ( U^1_x \psi )_{2n} = \psi_{2n}, \quad n \in \mathbb{Z}.
    \end{equation}
    Define operators $U^2_x$ which shift by $x$ in layer $2$ and leave layer $1$ unchanged analogously. Then, for each $x \in \mathbb{Z}$ and $b \in \mathbb{R}$, let 
    \begin{equation} \label{eq:T1T2}
        T^1_x b := b + x, \quad T^2_x b := b - (1-\theta) x. 
    \end{equation}
    Then, we have
    \begin{equation} \label{eq:layer_shifts}
        H( T^1_x b ) = U_x^1 H(b) ( U_x^1 )^\dagger, \quad H( T^2_x b ) = U_x^2 H(b) ( U_x^2 )^\dagger.
    \end{equation}
\end{lemma}
Identities \eqref{eq:layer_shifts} parallel \eqref{eq:unitary_operator_equivlence_0} in Definition \ref{def:standard_ergodic_operator_random_0}. We next require the following lemma, which establishes that the transformations \eqref{eq:T1T2} act ergodically when restricted to the unit cells of the other lattice with the uniform probability measure. For the proof, see e.g. \cite{Nadkarni1998}.
\begin{lemma}
    \label{lemma:ergodicity_chain_shift_transformations}
    Let $\tilde{T}^1_{x}$ and $\tilde{T}^2_{x}$ denote the transformations \eqref{eq:layer_shifts} restricted to the unit cells of layers $2$ and $1$ respectively, i.e.,
    \begin{equation} \label{eq:erg_trans}
        \begin{split}
            \tilde{T}^1_{x} b &= b + x \quad \quad \quad \quad \mod (1 - \theta),  \\
            \tilde{T}^2_{x} b &= b - (1-\theta) x \quad \mod 1
        \end{split}
    \end{equation}
    for each $b \in \mathbb{R}$ and $x \in \mathbb{Z}$. Then, if $0 < \theta < 1$ is irrational, the operators $( \tilde{T}^1_x )_{x \in \mathbb{Z}}$ form a group of measure-preserving transformations which acts ergodically on the probability space $([0,1-\theta),\mathcal{A},\mathbb{P}_2)$, where $\mathcal{A}$ denotes the Borel sets and $\mathbb{P}_2$ is the normalized Lebesgue measure on $[0,1-\theta)$. Similarly, the operators $( \tilde{T}^2_x )_{x \in \mathbb{Z}}$ act ergodically on the probability space $([0,1),\mathcal{A},\mathbb{P}_1)$, where $\mathbb{P}_1$ is the normalized Lebesgue measure on $[0,1)$. 
\end{lemma}
Note that the notation $\mathbb{P}_1$ for the measure on $[0,1)$ and $\mathbb{P}_2$ for the measure on $[0,1-\theta)$ is natural given that these are the unit cells of layers 1 and 2, respectively.

We can now define the limiting density of states measure concisely as follows.
\begin{definition}
    \label{def:dosm_incommensurate_chain_operator}
    The density of states measure for $H(b)$ is the measure $\text{d}k$ defined by the linear functional
    \begin{equation}
        g \mapsto \frac{1 - \theta}{1 + (1 - \theta)} \mathbb{E}_2[ \langle \delta_{10}, g(H(\cdot)) \delta_{10}\rangle ] + \frac{1}{1 + (1-\theta)} \mathbb{E}_1[ \langle \delta_{20}, g(H(\cdot)) \delta_{20}\rangle ] = \inty{\mathbb{R}}{}{ g }{k},
    \end{equation}
    where
    \begin{equation}
        \mathbb{E}_2[ \langle \delta_{10}, g(H(\cdot)) \delta_{10}\rangle ] = \inty{[0,1-\theta)}{}{ \langle \delta_{10}, g(H(b_2)) \delta_{10}\rangle }{\mathbb{P}_2} = \frac{1}{1 - \theta} \inty{[0,1-\theta)}{}{ \langle \delta_{10}, g(H(b_2)) \delta_{10}\rangle }{b_2},
    \end{equation}
    and 
    \begin{equation}
        \mathbb{E}_1[ \langle \delta_{20}, g(H(\cdot)) \delta_{20}\rangle ] = \inty{[0,1)}{}{ \langle \delta_{20}, g(H(b_1)) \delta_{20}\rangle }{\mathbb{P}_1} = \inty{[0,1)}{}{ \langle \delta_{20}, g(H(b_1)) \delta_{20}\rangle }{b_1}.
    \end{equation}
\end{definition}
In short, the limiting density of states measure is obtained by averaging the interlayer shift $b$ over the unit cells of each layer, up to normalization factors reflecting the excess of atoms in layer 2 compared with layer 1 within any interval $[-L,L]$.

We now state the theorem for the convergence of traces to the desired integral using Birkhoff's theorem for the incommensurate chain operators.
\begin{thm}
    \label{thm:convergence_dos_traces_weak_ergodic}
    Let $H(b)$ be the incommensurate chain operators defined in \eqref{eq:simple_hamiltonian_incommensurate_chain} equipped with the chain-shifting operators $U^1_{x}$ and $U^2_{x}$ in Lemma \ref{lemma:shifting_operators_incommensurate_operator_transformations}.
    If the difference $H_{\Lambda_L} - 1_{\Lambda_L}H(b)$ is trace-class and
    \begin{align}
        \operatorname{tr} \left| H_L(b) - 1_{\Lambda_{L}} H(b) \right| \leq \varepsilon(L)  |\Lambda_L|
        \label{eq:incommensurate_chain_op_truncation_trace_class_bound}
    \end{align} 
    with $\varepsilon(L) \to 0$ as $L \to \infty$, then
    \begin{multline}
         \frac{1}{|\Lambda_{L}|} \left( \sum_{x \in \Lambda_{L}^1} \langle \delta_{1x}, g(H_{L}(b)) \delta_{1x} \rangle + \sum_{x \in \Lambda_{L}^2} \langle \delta_{2x}, g(H_{L}(b)) \delta_{2x} \rangle \right)  \to \\ \frac{1-\theta}{1 + (1-\theta)} \int_{[0,1-\theta)} \langle \delta_{10}, g(H(b_{2})) \delta_{10}\rangle \,\text{d}\mathbb{P}_2   + \frac{1}{1+(1-\theta)} \int_{[0,1)} \langle \delta_{20}, g(H(b_{1})) \delta_{20}\rangle \,\text{d}\mathbb{P}_{1} \label{eq:result_convergence_incommensurate_chain_operators}
    \end{multline}
    for almost-all $b \in \mathbb{R}$ with respect to the Lebesgue measure and all $g \in C_c(\mathbb{R})$ as $L\to \infty$. In particular, we have that the measures $\text{d}k_{b,L}$ defined in \eqref{eq:finite_dosm_truncated_incommensurate_chain_operator} converge weakly to $\text{d}k$ as $L\to \infty$, i.e.
    \begin{align*}
        \lim_{L\to \infty} \int g \,\text{d}k_{b,L} = \int g \,\text{d}k.
    \end{align*}
\end{thm}

\begin{proof}[Proof of Theorem \ref{thm:convergence_dos_traces_weak_ergodic}]   
    The proof is similar to the proof of Theorem 3.15 of \cite{AizenmanWarzel2015}, but with subtle differences because the coupled chain Hamiltonian is not ergodic exactly in the sense of Definition \ref{def:standard_ergodic_operator_random_0}. We will first establish \eqref{eq:result_convergence_incommensurate_chain_operators} for $g(x) = (x-z)^{-1}$ with $z\in \mathbb{C}\setminus \mathbb{R}$. We will then extend to $g\in C_c(\mathbb{R})$ using Stone-Weierstrass. 


    First, note that for $g(x) = (x-z)^{-1}$ we have by a resolvent identity and \eqref{eq:incommensurate_chain_op_truncation_trace_class_bound} the estimate
    \begin{align}
        \left|\frac{1}{|\Lambda_{L}|} \sum_{x\in \Lambda_{L}} \langle \delta_{x}, g(H_{L}(b)) \delta_{x} \rangle - \frac{1}{|\Lambda_{L}|} \sum_{x\in \Lambda_{L}} \langle \delta_{x}, g(H(b)) \delta_{x} \rangle\right| \leq \frac{\varepsilon(L)}{|\Im(z)|^2}, \label{eq:uniform_estimate_incommensurate_chain_operator}
    \end{align}
    where we use the shorthand notation
    \begin{align*}
        \sum_{x\in \Lambda_{L}} \langle \delta_{x}, (\cdot) \delta_{x} \rangle = \sum_{x\in \Lambda_{L}^1} \langle \delta_{1x}, (\cdot) \delta_{1x} \rangle + \sum_{x\in \Lambda_{L}^2} \langle \delta_{2x}, (\cdot) \delta_{2x} \rangle.
    \end{align*}
    From this, we conclude
    \begin{equation}
        \lim_{L\to \infty} \left|\frac{1}{|\Lambda_{L}|} \sum_{x\in \Lambda_L} \langle \delta_{x}, g(H_L(b))\delta_{x}\rangle -  \frac{1}{|\Lambda_L|} \sum_{x\in \Lambda_L} \langle \delta_x, g(H(b))\delta_x\rangle \label{eq:limit_equivalence_truncations_incommensurate_chain} \right| = 0
    \end{equation}
    for any choice of $z\in \mathbb{C}\setminus \mathbb{R}$.

    Next, the second limit of \eqref{eq:limit_equivalence_truncations_incommensurate_chain} can be factored into a usable form for the application of Birkhoff's ergodic theorem: note we have that
    \begin{align}
        \frac{1}{|\Lambda_{L}|} \sum_{x\in \Lambda_L} \langle \delta_x, g(H(b))\delta_x\rangle & =
        \frac{|\Lambda_{L}^1|}{ |\Lambda_{L}^1|+|\Lambda_L^2|}  \left( \frac{1}{ |\Lambda_{L}^1|}  \sum_{x \in \Lambda_{L}^1} \langle \delta_{1x}, g(H(b))\delta_{1x}\rangle \right) \nonumber \\ & \qquad \qquad  +  \frac{|\Lambda_{L}^2|}{|\Lambda_{L}^1|+|\Lambda_L^2|} \left(\frac{1}{|\Lambda_{L}^2|} \sum_{x \in \Lambda_{L}^2} \langle \delta_{2x}, g(H(b))\delta_{2x}\rangle \right). \label{eq:limit_expansion_incommensurate_chain}
    \end{align}
    The terms in parentheses are the expressions we will use Birkhoff's ergodic theorem on. 

    First calculating the prefactors of the parentheses terms in \eqref{eq:limit_expansion_incommensurate_chain} above, we may simply calculate express the size of each truncation to each chain in \eqref{eq:size_of_truncations} as
    \begin{align*}
        |\Lambda_L^1| = 1 + 2L + \eta(L) \quad \text{and} \quad |\Lambda_L^2| = 1 + \frac{2L}{1-\theta} + \nu (L),
    \end{align*}
    where $\eta:\mathbb{R}\to \mathbb{R}$ and $\nu:\mathbb{R}\to \mathbb{R}$ are bounded functions of $L$.
    The limit of the size ratio can be calculated using these bounded functions to get
    \begin{align*}
        \lim_{L\to \infty} \frac{|\Lambda_{L}^1|}{|\Lambda_{L}^2|} = \lim_{L\to \infty} \frac{2 + \frac{1+ \eta(L)}{L}}{\frac{2}{1-\theta} + \frac{1 + \nu(L)}{L}} = 1-\theta,
    \end{align*}
    from which it can be derived that
    \begin{align}
        \lim_{L\to \infty}\frac{|\Lambda_{L}^1|}{ |\Lambda_{L}^1|+|\Lambda_L^2| } = \frac{1-\theta}{1 + (1-\theta)} \quad \text{and} \quad \lim_{L\to \infty}\frac{|\Lambda_{L}^2|}{ |\Lambda_{L}^1|+|\Lambda_L^2|} = \frac{1}{1+(1-\theta)}. \label{eq:limit_of_size_factors_incommensurate_chain}
    \end{align}
    This gives the prefactors of the summands on the right-hand side of \eqref{eq:limit_expansion_incommensurate_chain} in the limit as $L\to \infty$.

    Meanwhile, the limits of each summand in \eqref{eq:limit_expansion_incommensurate_chain} can be modified by shifting the $\delta_x$-vectors to a chosen fixed point, say to $\delta_{10}$ and $\delta_{20}$:
    \begin{align*}
        \sum_{x\in \Lambda_{L}^1} \langle \delta_{1x}, g(H(b))\delta_{1x}\rangle & = \sum_{x\in \Lambda_{L}^1} \langle U_x^1 \delta_{10}, g(H(b)) U^1_x \delta_{10}\rangle \\
        & = \sum_{x\in \Lambda_{L}^1} \langle  \delta_{10}, (U^1_x)^\dagger g(H(b)) U^1_x \delta_{10}\rangle,
    \end{align*}
    where $U^i_x$ is defined in Lemma \ref{lemma:shifting_operators_incommensurate_operator_transformations} and satisfies $(U^i_x)^\dagger = U^i_{-x}$ for each $x \in \mathbb{Z}$.
    Furthermore, since the $\delta_{1x}$ (any $x \in \mathbb{Z}$) vectors are fixed under $U_y^2$ for all $y \in \mathbb{Z}$, we have
    \begin{align}
        \sum_{x\in \Lambda_{L}^1} \langle \delta_{1x}, g(H(b))\delta_{1x}\rangle = \sum_{x\in \Lambda_{L}^1} \langle  \delta_{10}, ( U^2_y U^1_x)^\dagger g(H(b)) U^1_x U^2_y \delta_{10}\rangle \label{eq:periodicity_wrt_other_chain}
    \end{align}
    for any $y \in \mathbb{Z}$. In particular, note that for any $b\in \mathbb{R}$ and any $x \in \mathbb{Z}$, there is an integer $[b-x]_{1-\theta} \in \mathbb{Z}$ such that $b - x + (1-\theta)[b-x]_{1-\theta} \in [0,1-\theta)$. We now set $y$ in \eqref{eq:periodicity_wrt_other_chain} to equal this $[b-x]_{1-\theta}$. Since $g$ is a resolvent function, 
    we get 
    \begin{align} \label{eq:first_chain_periodicity_incommensurate_chain}
        \sum_{x\in \Lambda_{L}^1} \langle \delta_{1x}, g(H(b))\delta_{1x}\rangle & = \sum_{x\in \Lambda_{L}^1} \langle  \delta_{10},  g\left((U^1_x U^2_{[b-x]_{1-\theta}})^\dagger H(b)U^1_x U^2_{[b-x]_{1-\theta}}\right)  \delta_{10}\rangle \nonumber\\&  = \sum_{x\in \Lambda_{L}^1} \langle \delta_{10}, g(H(\tilde{T}^1_x b))\delta_{10}\rangle.
    \end{align}
    Note that because we chose $y$ to equal $[b-x]_{1-\theta}$, in the argument of $H$ we can replace $T_x^1 b$, which appeared in \eqref{eq:T1T2}, by $\tilde{T}_x^1 b$, the ergodic transformation of the unit cell of layer 2 appearing in Lemma \ref{lemma:ergodicity_chain_shift_transformations}.
    Likewise,
    \begin{align}
        \sum_{x\in \Lambda_{L}^2} \langle \delta_{2x}, g(H(b))\delta_{2x}\rangle = \sum_{x\in \Lambda_{L}^2} \langle \delta_{20}, g(H(\tilde{T}^2_x b))\delta_{20}\rangle \label{eq:second_chain_periodicity_incommensurate_chain}
    \end{align}
    for the summation over the second chain.

    We may now combine \eqref{eq:limit_of_size_factors_incommensurate_chain}, \eqref{eq:first_chain_periodicity_incommensurate_chain}, and \eqref{eq:second_chain_periodicity_incommensurate_chain} into  \eqref{eq:limit_expansion_incommensurate_chain}, getting
    \begin{align*}
        \lim_{L\to \infty} \frac{1}{|\Lambda_L|} \sum_{x\in \Lambda_L} \langle \delta_x, g(H(b))\delta_x\rangle& = \frac{1-\theta}{1 + (1-\theta)} \left( \lim_{L\to \infty} \frac{1}{|\Lambda_L^1|} \sum_{x\in \Lambda_{L}^1} \langle \delta_{10}, g(H(\tilde{T}^1_x b))\delta_{10}\rangle\right) \\
        & \qquad \qquad + \frac{1}{1 + (1-\theta)} \left(\lim_{L\to \infty} \frac{1}{|\Lambda_L^2|} \sum_{x\in \Lambda_{L}^2} \langle \delta_{20}, g(H(\tilde{T}^2_x b))\delta_{20}\rangle\right).
    \end{align*}
    Let $\tilde{b}_1 \in [0,1-\theta)$ be such that $\tilde{b}_1 = b \mod (1-\theta)\mathbb{Z}$, and $\tilde{b}_2 \in [0,1)$ be such that $\tilde{b}_2 = b \mod \mathbb{Z}$. We can now apply Birkhoff's ergodic theorem to each limit in the above expression to obtain, together with \eqref{eq:limit_equivalence_truncations_incommensurate_chain},
        \begin{align}
        \lim_{L\to \infty} \frac{1}{|\Lambda_L|} \sum_{x\in \Lambda_L} \langle \delta_x, g(H_L(b))\delta_x\rangle& = \frac{1-\theta}{1 + (1-\theta)} \int_{[0,1-\theta)} \langle \delta_{10}, g(H(b_{2})) \delta_{10}\rangle \,\text{d}\mathbb{P}_2 \nonumber \\
        &\qquad \qquad + \frac{1}{1 + (1-\theta)} \int_{[0,1)} \langle \delta_{20}, g(H(b_{1})) \delta_{20}\rangle \,\text{d}\mathbb{P}_{1}, \label{eq:pre_result_convergence_incommensurate_chain}
    \end{align}
    for $g(x) = (x - z)^{-1}$ and full measure sets $\tilde{b}_i \in \Omega^z_i$, $i \in \{1,2\}$. By taking the countable intersection of these sets over the set of $z \in \mathbb{C} \setminus \mathbb{R}$ with rational coefficients, we have convergence in the still full measure sets $\tilde{\Omega}_i$, $i \in \{1,2\}$, for a set of functions which is dense in $C_c(\mathbb{R})$ in the uniform topology, and hence for all $f \in C_c(\mathbb{R})$. We thus obtain convergence for all $b$ in the complement of the measure $0$ set given by the union of the sets $[0,1-\theta) \setminus \tilde{\Omega}_1$ and $[0,1) \setminus \tilde{\Omega}_2$ with their shifts by $(1-\theta)\mathbb{Z}$ and $\mathbb{Z}$.

\end{proof}


\remark{\label{remark:direct_continuous_proof_incommensurate_dosm_convergence}For a proof of Birkhoff's ergodic theorem, see \cite{Krengel1985}. For smooth functions a simpler proof by expanding the function in a Fourier series is available; see e.g. \cite{Massatt2017}.}


\subsection{Almost-sure spectrum for the incommensurate coupled chain model}
\label{section:almost-sure_spectrum_iccm}

In this section, we prove the following analog of Pastur's theorem for the full coupled chain model.
\begin{thm}
    \label{thm:almost-sure_spectrum_iccm}
    Let $H(b)$ be the incommensurate coupled chain operators in Definition \ref{def:incommensurate_coupled_chain_hamiltonian} with $\theta \in (0,1)$ irrational. Then there exists a set $\Sigma \subset \mathbb{R}$, called the \emph{almost-sure spectrum} of $H(b)$, such that $\sigma(H(b)) = \Sigma$ for almost all $b \in \mathbb{R}$.
\end{thm}
\begin{proof}
    The following proof parallels that of Theorem 3.10 of \cite{AizenmanWarzel2015}. For any $E_1,E_2\in \mathbb{R}$, the functions 
    \begin{align}
        b \to X_{E_1,E_2}(b) \vcentcolon= \operatorname{dim}\operatorname{range}1_{(E_1,E_2)}(H(b))
    \end{align}
    defined by the sum of nonnegative terms $X_{E_1,E_2}(b) = \sum_{i=1}^\infty \langle \psi_k, 1_{(E_1,E_2)}(H(b))\psi_k \rangle$ for any orthonormal basis $(\psi_k)_{k \in \mathbb{Z}^+}$ of the coupled chain Hilbert space $\mathcal{H}$, are measurable, taking values in $\mathbb{N}\cup \{+\infty\}$. Since the sum is independent of the choice of basis, we can choose WLOG $(\psi_k)_{k \in \mathbb{Z}^+} = (\delta_{1i})_{i\in \mathbb{Z}} \cup (\delta_{2j})_{j\in \mathbb{Z}}$ of $\mathcal{H}$ as in Definition \ref{def:truncated_dosm_iccm}, so that
    \begin{equation}
        X_{E_1,E_2}(b) = X^1_{E_1,E_2}(\tilde{b}_1) + X^2_{E_1,E_2}(\tilde{b}_2), \quad X^i_{E_1,E_2}(b) := \sum_{x\in \mathbb{Z}}\langle \delta_{ix}, 1_{(E_1,E_2)}(H(b)) \delta_{ix} \rangle, 
    \end{equation}
    where $\tilde{b}_1 \in [0,1-\theta)$ is such that $\tilde{b}_1 = b \mod (1-\theta)\mathbb{Z}$ and $\tilde{b}_2 \in [0,1)$ be such that $\tilde{b}_2 = b \mod \mathbb{Z}$. Now note that the action of the translation operators \eqref{eq:erg_trans} leaves the functions $X^i_{E_1,E_2}(\cdot)$ invariant by applying the unitary transformations in Lemma \ref{lemma:shifting_operators_incommensurate_operator_transformations}, and hence, by ergodicity (Exercise 3.1 of \cite{AizenmanWarzel2015}), we have that there exist numbers $c^1_{E_1,E_2}, c^2_{E_1,E_2} \in [0,\infty]$ such that
    \begin{equation}
        X^1_{E_1,E_2}(\tilde{b}_1) = c^1_{E_1,E_2}, \quad X^2_{E_1,E_2}(\tilde{b}_2) = c^2_{E_1,E_2}
    \end{equation}
    with probability $1$ in either case. It follows then that there exist numbers $c_{E_1,E_2} \in [0,\infty]$ such that
    \begin{equation} \label{eq:hold}
        X_{E_1,E_2}(b) = c_{E_1,E_2}
    \end{equation}
    for all $b \in \mathbb{R}$ in a full measure set, and moreover that the set of $b$ such that
    \begin{equation} \label{eq:full_meas}
        X_{E_1,E_2}(b) = c_{E_1,E_2} \text{ for all } E_1, E_2 \in \mathbb{Q} 
    \end{equation}
    still has full measure, being the countable intersection of full measure sets. But now recall the characterization of the spectrum in terms of spectral projections
    \begin{equation}
        \sigma(H(b)) = \left\{ E \in \mathbb{R} : \forall \epsilon > 0, \operatorname{dim}\operatorname{range}1_{(E-\epsilon,E+\epsilon)}(H(b)) > 0 \right\}.
    \end{equation}
    For all $b$ in the full measure set such that \eqref{eq:full_meas} holds, we can characterize the right-hand side using the $c_{E_1,E_2}$ resulting in 
    \begin{equation}
        \sigma(H(b)) = \Sigma := \left\{ E \in \mathbb{R} : c_{E_1,E_2} > 0 \text{ for all } E_1, E_2 \in \mathbb{Q} \text{ with } E_1 < E < E_2 \right\}
    \end{equation}
    for all $b \in \mathbb{R}$ in a full measure set. 
\end{proof}

\subsection{Convergence of the density of states measure for random incommensurate coupled chain model}
In this section, we define the density of states for the random incommensurate coupled chain model and show convergence from truncations of the operator as in Section \ref{sec:convergence_dosm_iccm}.

First, we similarly define the density of states for the random incommensurate coupled chain truncated to the interval $[-L,L]$ for $L>0$ on each chain with Dirichlet boundary conditions.

\begin{definition}
    For positive $L>0$, let $\Lambda_L^1$, $\Lambda_L^2$, $|\Lambda_L|$, and $1_{\Lambda_L}$ be defined as in Definition \ref{def:truncated_dosm_iccm}. 
    Then, let $H_{R,L}(b,\omega_1,\omega_2) \vcentcolon= 1_{\Lambda_L} H_R(b,\omega_1,\omega_2) 1_{\Lambda_L}$. 
    For each fixed $b\in \mathbb{R}$ and $\omega_1,\omega_2 \in \mathbb{R}^\mathbb{Z}$, and $L>0$, we define the density of states measure for $H_{R,L}(b,\omega_1,\omega_2)$ as the unique Radon measure $dk_{b,\omega_1,\omega_2,L}$ corresponding to the linear functional
    \begin{align}
        g\mapsto \frac{1}{|\Lambda_L|} \left( \sum_{x \in \Lambda_L^1} \braket{\delta_{1x}, g(H_{R,L}(b,\omega_1,\omega_2))\delta_{1x} }   + \sum_{x \in \Lambda_L^2} \braket{\delta_{2x}, g(H_{R,L}(b,\omega_1,\omega_2))\delta_{2x} }  \right) = \int_{\mathbb{R}} g\, \text{d}k_{b,\omega_1,\omega_2,L},
    \end{align}
    where $\delta_{ix}$, $i \in \{1,2\}$, are as in Definition \ref{def:truncated_dosm_iccm}.
\end{definition}

We now define the ergodic structure necessary for applying Birkhoff's ergodic theorem for proving the convergence of the density states measure in the large truncation limit. 

\begin{lemma} \label{lemma:random_iccm_covariance_shifting}
    For each $x\in \mathbb{Z}$, $b\in \mathbb{R}$, and $\omega_1,\omega_2 \in \mathbb{R}^\mathbb{Z}$, let $T_{R,x}^1,T_{R,x}^2:\mathbb{R} \times \mathbb{R}^\mathbb{Z}\times\mathbb{R}^\mathbb{Z} \to \mathbb{R} \times \mathbb{R}^\mathbb{Z}\times \mathbb{R}^\mathbb{Z}$ be defined with
    \begin{align}
            T_{R,x}^1(b, \omega_{1}, \omega_{2}) &= (b+x, (\omega_{1(j-x)})_{j \in \mathbb{Z}}, \omega_{2} ) , \\ T_{R,x}^2(b, \omega_{1}, \omega_{2}) &= (b-(1-\theta)x, \omega_{1}, (\omega_{2(j-x)})_{j \in \mathbb{Z}} ).
    \end{align}
    Then, 
    \begin{align}
        H(T_{R,x}^1(b,\omega_1,\omega_2)) = U_x^1 H(b,\omega_1,\omega_2) (U_x^1)^\dagger, \qquad H(T_{R,x}^2(b,\omega_1,\omega_2)) = U_x^2 H(b,\omega_1,\omega_2) (U_x^2)^\dagger,
    \end{align}
    where $U_x^1$ and $U_x^2$ are the right-shift operators on each chain defined as in Lemma \ref{lemma:shifting_operators_incommensurate_operator_transformations}. 
\end{lemma}

We next introduce and establish the ergodicity of the transformations in Lemma \ref{lemma:random_iccm_covariance_shifting} when acting on each chain individually.

\begin{definition} \label{def:ergodic_transf_random_iccm}  
    For each $x \in \mathbb{Z}$, let $\tilde{T}_{R,x}^1$ and $\tilde{T}_{R,x}^2$ denote the transformations given by
    \begin{align}
        \tilde{T}_{R,x}^1(b,\omega_1,\omega_2)  & = \left(b+x - (1-\theta) [b+x]_{1-\theta}, (\omega_{1(j-x)})_{j \in \mathbb{Z}}, (\omega_{2(j - [b+x]_{1-\theta})})_{j \in \mathbb{Z}} \right) \\
        \tilde{T}_{R,x}^2(b,\omega_1,\omega_2)  & = \left(b-(1-\theta)x + [b+x]_{1} , (\omega_{1(j-[b+x]_{1})})_{j\in \mathbb{Z}}, (\omega_{2(j-x)})_{j \in \mathbb{Z}} \right)
    \end{align}
    on the product probability spaces $([0,1-\theta) \times \mathbb{R}^\mathbb{Z} \times \mathbb{R}^\mathbb{Z}, \mathcal{A}_2, \mathbb{P}_{R,2} )$ and $([0,1) \times \mathbb{R}^\mathbb{Z}\times \mathbb{R}^\mathbb{Z}, \mathcal{A}_1, \mathbb{P}_{R,1})$, respectively,
    where $[b+x]_{1-\theta}$ and $[b+x]_{1}$ are the unique integers such that
    \begin{equation}
        b + x - [b+x]_{1-\theta} (1 - \theta) \in [0,1-\theta), \quad b + x - [b+x]_{1} \in [0,1).
    \end{equation}
\end{definition}

\begin{lemma} \label{lemma:random_iccm_ergodic_shifting_incom_unit_cells}
    The transformations $\{\tilde{T}^1_{R,x}\}_{x\in \mathbb{Z}}$ and $\{\tilde{T}^2_{R,x}\}_{x\in \mathbb{Z}}$ in Definition \ref{def:ergodic_transf_random_iccm} form groups which act ergodically on $([0,1-\theta) \times \mathbb{R}^\mathbb{Z} \times \mathbb{R}^\mathbb{Z}, \mathcal{A}_2, \mathbb{P}_{R,2})$ and $([0,1) \times \mathbb{R}^\mathbb{Z} \times \mathbb{R}^\mathbb{Z}, \mathcal{A}_2, \mathbb{P}_{R,1})$, respectively.
\end{lemma}
A proof of Lemma \ref{lemma:random_iccm_ergodic_shifting_incom_unit_cells} is contained in Appendix \ref{sec:proofs_random_iccm}.

We now give a theorem for the convergences of the density of states for the random incommensurate chain operators.
Most details of the proof are parallel to that of Theorem \ref{thm:convergence_dos_traces_weak_ergodic}.
\begin{definition}
    \label{def:dosm_random_iccm}
    The density of states measure for $H_R(b,\omega_1,\omega_2)$ is the measure $\text{d}k$ from the linear functional
    \begin{align}
        g \mapsto \frac{1 - \theta}{1 + (1 - \theta)} \mathbb{E}_{2,R}[ \langle \delta_{10}, g(H_R(\cdot)) \delta_{10}\rangle ] + \frac{1}{1 + (1-\theta)} \mathbb{E}_{1,R}[ \langle \delta_{20}, g(H_R(\cdot)) \delta_{20}\rangle ] = \inty{\mathbb{R}}{}{ g }{k},\label{eq:dosm_random_iccm}
    \end{align}
    where 
    \begin{align}
        \mathbb{E}_{2,R}[ \langle \delta_{10}, g(H_R(\cdot)) \delta_{10}\rangle ] = \int_{[0,1-\theta)\times\mathbb{R}^\mathbb{Z}\times\mathbb{R}^\mathbb{Z}} \braket{\delta_{10}, g(H_R(b,\omega_1,\omega_2)) \delta_{10}} \,\text{d}\mathbb{P}_{R,2}
    \end{align}
    and
    \begin{align}
        \mathbb{E}_{1,R}[ \langle \delta_{20}, g(H_R(\cdot)) \delta_{20}\rangle ] = \int_{[0,1)\times\mathbb{R}^\mathbb{Z}\times\mathbb{R}^\mathbb{Z}}\braket{\delta_{20}, g(H_R(b,\omega_1,\omega_2)) \delta_{20}} \,\text{d}\mathbb{P}_{R,1}.
    \end{align}
\end{definition}

\begin{thm}
    \label{thm:dos_convergence_random_iccm}
    Let $H_{R}(b,\omega_1,\omega_2)$ be the random incommensurate operator in Definition \ref{def:random_iccm_operator} with $\theta \in (0,1)$ irrational.
    If the difference $H_{R,L}(b,\omega_1,\omega_2) - 1_{\Lambda_L}H_R(b,\omega_1,\omega_2)$ is trace-class and
    \begin{align}
        \operatorname{tr}\left| H_{R,L}(b,\omega_1,\omega_2) - 1_{\Lambda_L}H_R(b,\omega_1,\omega_2) \right| \leq \epsilon(L) |\Lambda_L| \label{eq:trace_bound_dos_convergence_random_iccm}
    \end{align}
    with $\epsilon(L) \to 0$ as $L\to \infty$, then
    \begin{align}
        \frac{1}{|\Lambda_L|}  \sum_{x \in \Lambda_L} \braket{\delta_{x}, g(H_{R,L}(b,\omega_1,\omega_2)) \delta_{x}}  \to &  \,\frac{1 - \theta}{1 + (1-\theta)} \int_{[0,1-\theta)\times\mathbb{R}^\mathbb{Z}\times\mathbb{R}^\mathbb{Z}} \braket{\delta_{0_1}, g(H_R(b,\omega_1,\omega_2)) \delta_{0_1}} \,\text{d}\mathbb{P}_{R,2} \nonumber \\   
        &\quad +   \frac{1}{1+(1-\theta)} \int_{[0,1)\times\mathbb{R}^\mathbb{Z}\times\mathbb{R}^\mathbb{Z}}\braket{\delta_{0_2}, g(H_R(b,\omega_1,\omega_2)) \delta_{0_2}} \,\text{d}\mathbb{P}_{R,1} \label{eq:dos_convergence_random_iccm}
    \end{align}
    for all $g\in C_0(\mathbb{R})$ and almost-all $(b,\omega_1,\omega_2) \in \mathbb{R} \times \mathbb{R}^\mathbb{Z} \times \mathbb{R}^\mathbb{Z}$ with respect to the product measure as $L \to \infty$.
\end{thm}
\begin{proof}
    We provide a similar proof to that of Theorem \ref{thm:convergence_dos_traces_weak_ergodic}, establishing \eqref{eq:dos_convergence_random_iccm} first for $g(x) = (x-z)^{-1}$ with $z\in \mathbb{C}\setminus\mathbb{R}$ and extending to more general functions using Stone-Weierstrass.
    
    First, the trace-norm bound of \eqref{eq:trace_bound_dos_convergence_random_iccm} gives the estimate 
    \begin{align}
        \left| \frac{1}{|\Lambda_L|} \sum_{x\in \Lambda_L} \braket{\delta_{x}, g(H_{R,L}(b,\omega_1,\omega_2))\delta_x} - \frac{1}{|\Lambda_L|} \sum_{x\in \Lambda_L} \braket{\delta_{x}, g(H_{R,L}(b,\omega_1,\omega_2))\delta_x} \right| \leq \frac{\epsilon(L)}{|\operatorname{Im}(z)|^2}
    \end{align}
    on such resolvent functions $g$.
    Hence, for any $z\in \mathbb{C}\setminus \mathbb{R}$
    \begin{align}
        \lim_{L\to \infty} \left| \frac{1}{|\Lambda_L|} \sum_{x\in \Lambda_L} \braket{\delta_{x}, g(H_{R,L}(b,\omega_1,\omega_2))\delta_x} - \frac{1}{|\Lambda_L|} \sum_{x\in \Lambda_L} \braket{\delta_{x}, g(H_{R,L}(b,\omega_1,\omega_2))\delta_x} \right| = 0. \label{eq:limit_truncation_approximation_convergence_dos_random_iccm}
    \end{align}
    Recall lastly that we decompose the summation over the untruncated operator by 
    \begin{align}
        \frac{1}{|\Lambda_{L}|} \sum_{x\in \Lambda_L} \langle \delta_x, g(H_R(b,\omega_1,\omega_2))\delta_x\rangle & =
        \frac{|\Lambda_{L}^1|}{ |\Lambda_{L}^1|+|\Lambda_L^2|}  \left( \frac{1}{ |\Lambda_{L}^1|}  \sum_{x_1\in \Lambda_{L}^1} \langle \delta_{x_1}, g(H_R(b,\omega_1,\omega_2))\delta_{x_1}\rangle \right) \nonumber \\ & \qquad \qquad  +  \frac{|\Lambda_{L}^2|}{|\Lambda_{L}^1|+|\Lambda_L^2|} \left(\frac{1}{|\Lambda_{L}^2|} \sum_{x_2\in \Lambda_{L}^2} \langle \delta_{x_2}, g(H_R(b,\omega_1,\omega_2))\delta_{x_2}\rangle \right),\label{eq:limit_expansion_convergence_dos_random_iccm}
    \end{align}
    where the limiting quantity of the prefactors are given by \eqref{eq:limit_of_size_factors_incommensurate_chain}.

    We next utilize the shift-covariance relationship to apply Birkhoff's theorem on each chain. For the summation over the first chain's sites in \eqref{eq:limit_expansion_convergence_dos_random_iccm}, 
    \begin{align}
        \sum_{x \in \Lambda^1_L} \braket{ \delta_{1x}, g(H_R(b,\omega_1,\omega_2))\delta_{1x} } = \sum_{x\in \Lambda^1_L} \braket{\delta_{10}, (U^1_x)^\dagger g(H_R(b,\omega_1,\omega_2)) U^1_x \delta_{10} } \nonumber 
    \end{align}
    via the definitions in Lemma \ref{lemma:random_iccm_covariance_shifting}. By the shift invariance of $\delta_{1x}$ with respect to shifts of the opposite chain, 
    \begin{align}
        \sum_{x \in \Lambda^1_L} \braket{ \delta_{1x}, g(H_R(b,\omega_1,\omega_2))\delta_{1x} } = & \sum_{x\in \Lambda^1_L} \braket{\delta_{10}, ( U^2_{[x+b]_{1-\theta}} U^1_x)^\dagger g(H_R(b,\omega_1,\omega_2)) U^1_x U^2_{[x+b]_{1-\theta}}\delta_{10} } \nonumber \\
        & = \sum_{x\in \Lambda^1_L} \braket{\delta_{10},  g(H_R(\tilde{T}^1_{x}(b,\omega_1,\omega_2))) \delta_{10} } \label{eq:1_chain_transformation_before_lim_convergence_dos_random_iccm}
    \end{align}
    from Definition \ref{def:ergodic_transf_random_iccm}. Likewise,
    \begin{align}
        \sum_{x \in \Lambda^1_L} \braket{ \delta_{2x}, g(H_R(b,\omega_1,\omega_2))\delta_{2x} }  
        & = \sum_{x\in \Lambda^1_L} \braket{\delta_{20},  g(H_R(\tilde{T}^2_{x}(b,\omega_1,\omega_2))) \delta_{20} }. \label{eq:1-theta_chain_transformation_before_lim_convergence_dos_random_iccm}
    \end{align}

    Next, combining \eqref{eq:limit_expansion_convergence_dos_random_iccm}, \eqref{eq:1_chain_transformation_before_lim_convergence_dos_random_iccm}, \eqref{eq:1-theta_chain_transformation_before_lim_convergence_dos_random_iccm}, and \eqref{eq:limit_of_size_factors_incommensurate_chain}, 
    \begin{align}
        \lim_{L\to \infty} \frac{1}{|\Lambda_L|} \sum_{x\in \Lambda_L} \braket{\delta_x, g(H_R(b,\omega_1,\omega_2))\delta_x} & = \frac{1-\theta}{1+(1-\theta)}  \left( \lim_{L\to \infty} \frac{1}{ |\Lambda_{L}^1|}  \sum_{x\in \Lambda_{L}^1} \langle \delta_{10}, g(H_R(\tilde{T}^1_{x}(b,\omega_1,\omega_2)))\delta_{10}\rangle \right) \nonumber  \\    & \quad  +  \frac{1}{1+(1-\theta)} \left(\lim_{L\to \infty}\frac{1}{|\Lambda_{L}^2|} \sum_{x\in \Lambda_{L}^2} \langle \delta_{20}, g(H_R(\tilde{T}^2_x(b,\omega_1,\omega_2)))\delta_{20}\rangle \right). \nonumber 
    \end{align}
    We can now apply Birkhoff's ergodic theorem for each $z \in \mathbb{C}\setminus \mathbb{R}$ to see that 
    \begin{align}
        \lim_{L\to \infty} \frac{1}{|\Lambda_L|} \sum_{x\in \Lambda_L} \braket{\delta_x, g(H_R(b,\omega_1,\omega_2))\delta_x} & = \frac{1-\theta}{1+(1-\theta)}  \int_{[0,1-\theta)\times \mathbb{R}^\mathbb{Z}\times\mathbb{R}^\mathbb{Z}} \langle \delta_{10}, g(H_R(\tilde{T}^1_{x}(b,\omega_1,\omega_2)))\delta_{10}\rangle \,\text{d}\mathbb{P}_{R,2} \nonumber  \\      +\,&  \frac{1}{1+(1-\theta)} \int_{[0,1)\times \mathbb{R}^\mathbb{Z}\times \mathbb{R}^\mathbb{Z}}\langle \delta_{20}, g(H_R(\tilde{T}^2_x(b,\omega_1,\omega_2)))\delta_{20}\rangle \,\text{d}\mathbb{P}_{R,1} \label{eq:pre-result_convergence_dos_random_iccm}
    \end{align}
    when $g(x) = (x-z)^{-1}$ for $(b,\omega_1,\omega_2)$ in a full measure set in $\mathbb{R}\times \mathbb{R}^\mathbb{Z}\times\mathbb{R}^\mathbb{Z}$, similar to Theorem \ref{thm:convergence_dos_traces_weak_ergodic}. The extension to general $g \in C_c(\mathbb{R})$ using Stone-Weierstrass is just as in the proof of Theorem \ref{thm:dos_convergence_random_iccm}.
\end{proof}

\subsection{Almost-sure spectrum for random incommensurate operators}

In this section, we establish the almost-sure spectrum for the incommensurate coupled chain model with on-site disorder.
The mathematical statements and content are similar to that of Section \ref{section:almost-sure_spectrum_iccm}.

\begin{thm}
    Let $H_R(b,\omega_1,\omega_2)$ be the random incommensurate coupled chain operators in Definition \ref{def:random_iccm_operator} with $\theta \in (0,1)$ irrational.
    Then there exists a set $\Sigma_R \subset \mathbb{R}$, called the \emph{almost-sure spectrum} of $H_R(b,\omega_1,\omega_2)$, such that the condition 
    \begin{align}
        \left\{ (b,\omega_1,\omega_2) \in \mathbb{R}\times \mathbb{R}^\mathbb{Z} \times \mathbb{R}^\mathbb{Z}: \sigma(H_R(b,\omega_1,\omega_2)) =\Sigma_R \right\}
    \end{align}
    has full-Lebesgue measure in $\mathbb{R}\times \mathbb{R}^\mathbb{Z} \times \mathbb{R}^\mathbb{Z}$.
\end{thm}
\begin{proof}
    The proof is similar to that of Theorem \ref{thm:almost-sure_spectrum_iccm}. 
    The only notable difference is in the calculation paralleling \eqref{eq:1_chain_transformation_before_lim_convergence_dos_random_iccm} and \eqref{eq:1-theta_chain_transformation_before_lim_convergence_dos_random_iccm} through applying Lemmas \ref{lemma:random_iccm_covariance_shifting} and \ref{lemma:random_iccm_ergodic_shifting_incom_unit_cells} instead of \ref{lemma:shifting_operators_incommensurate_operator_transformations} and \ref{lemma:ergodicity_chain_shift_transformations}.
\end{proof}

\section{Numerical Results} \label{sec:numerics}

\subsection{Numerical Method}

In this section we discuss computation of the limiting densities of states derived in the previous sections. In order to numerically compute \eqref{eq:red_this}, \eqref{eq:pre_result_convergence_incommensurate_chain}, and \eqref{eq:pre-result_convergence_dos_random_iccm} efficiently and to high accuracy, we used the popular kernel polynomial method (KPM). This method uses Chebyshev polynomials for quick calculations of diagonal elements of functions of Hamiltonians $H$, which we write here simply as $\langle \delta_x, g(H) \delta_x \rangle$. The details of the method can be found in \cite{Weisse2006} and the particular algorithm used can be found in \cite{Massatt2017}. We provide a brief recap of this method for the convenience of the reader in Appendix \ref{sec:KPM}. 
We then approximated the integrals in \eqref{eq:red_this}, \eqref{eq:pre_result_convergence_incommensurate_chain}, and \eqref{eq:pre-result_convergence_dos_random_iccm} with respect to disregistry with Riemann sums.

In order to apply $g$ to our different Hamiltonians numerically, we had to truncate the Hamiltonian. 
Since the density of states integral only requires evaluation at one site, we chose that site as the center of the cell and only expanded the Hamiltonian out to the nearest $R$ neighbors in each direction. The middle sites were chosen to be aligned at zero disregistry ($b = 0$). In order to compare wide ranges of different interatomic spacings, the interatomic distances were chosen to be $\sqrt{q/p}$ and $\sqrt{p/q}$ for the top chain and bottom chain respectively, for integers $p, q$. This creates an interatomic ratio of $p/q$. Note that exponential convergence can be expected with respect to $R$ for real-analytic $g$ using Combes-Thomas estimates on the resolvent \cite{1973CombesThomas}.
To calculate the plots seen below, we generated a list of $p$ and $q$ values and the associated matrices corresponding to these values and then used the kernel polynomial method to find the density of states at each desired energy value. The list of $p$ and $q$ values was determined by setting a maximal value $M$ and then creating a list of values that summed to $M$ such that $p+q = M$. The minimum $p$ value was chosen to be $M/7$.

\begin{figure}
    \centering
    \includegraphics[width=0.8\linewidth]{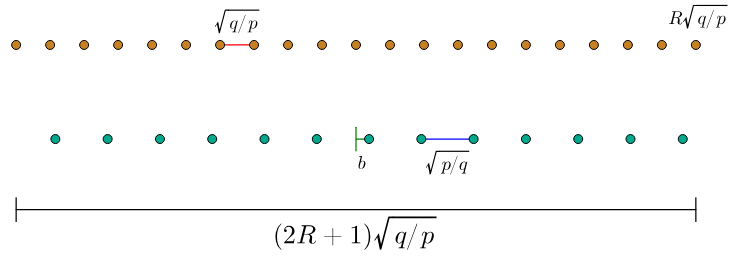}
    \caption{A small $R=10$ demonstration of how the coupled chains are arranged. The center atoms in the chain are aligned at $b=0$ and $b>0$ shifts the entire lower chain.}
    \label{fig:chaindemo}
\end{figure} 

\subsection{Numerical results}

We show in Figure \ref{fig:rcnumerics_Lap} the computed density of states for the reduced incommensurate chain model \eqref{eq:reduced_single_chain} using the integral formula \eqref{eq:red_this}, and in Figures \ref{fig:rcnumerics} and \ref{fig:rczoom} computations of the density of states for the model \eqref{eq:reduced_single_chain_noLap} where the inverse Laplacian is replaced by identity. The parameters used are shown in the following tables. Here $p_{\text{min}}$ and $p_{\text{max}}$ are the minimum and maximum $p$ values, where $p, M, R$ are as in the previous section, $N, \eta$ are the maximum degree of Chebyshev polynomial used and rescaling constant, respectively (see Appendix \ref{sec:KPM}), $\delta E$ is the energy discretization, and $B$ is the number of grid points used to evaluate the integral over $b$. We chose $h(x)$ to have the form as in \eqref{eq:inter_hop} with $A= 2$, $B = 5$, and $L = .1$. We chose $\mu$ as 10, $A=\sqrt{\mu}$ so that the interlayer hopping has the same magnitude as the intralayer term. 
\begin{center}
\begin{tabular}{|c|c|c|c|c|c|c|c|}
     \hline
     $M$ & $p_{\text{min}}$ & $p_{\text{max}}$ & $N$ & $\eta$ & $\Delta E$ & $R$ & $B$ \\
     \hline
     1000&142&858&200&$1/5$&.01&50 & 100\\ 
     \hline
\end{tabular} \\ \vspace{1mm} Parameter table for Figures \ref{fig:ccnumerics}, \ref{fig:rcnumerics} and \ref{fig:onsitenumerics}
\end{center}
\begin{center}
\begin{tabular}{|c|c|c|c|c|c|c|c|}
     \hline
     $M$ & $p_{\text{min}}$ & $p_{\text{max}}$ & $N$ & $\eta$ & $\Delta E$ & $R$ & $B$ \\
     \hline
     1000&14&299&500&$1/5$&.01&50 & 100\\ 
     \hline
\end{tabular}\\ \vspace{1mm} Parameter table for Figure \ref{fig:rczoom}
\end{center}
The computations show a complex fractal-like structure reminiscent of the Hofstadter butterfly \cite{PhysRevB.14.2239}. We see that the energy of the system increases as the ratio gets smaller. As we proceed left in the plot, the $q$ parameter gets larger relative to the $p$ parameter. This corresponds with the spacing between sites in the second layer increasing, while sites in the first layer, with interatomic distance $\sqrt{q/p}$, get further apart. This results in the second layer forming an onsite potential at the central site. The model effectively becomes a discrete Laplacian (with no onsite term since we dropped the onsite term in our numerics) but adding a single identity term. This shifts the base spectrum of $[-2,2]$ upwards to $[-1,3]$.

\begin{figure}
    \centering
    \includegraphics[width=0.75\linewidth]{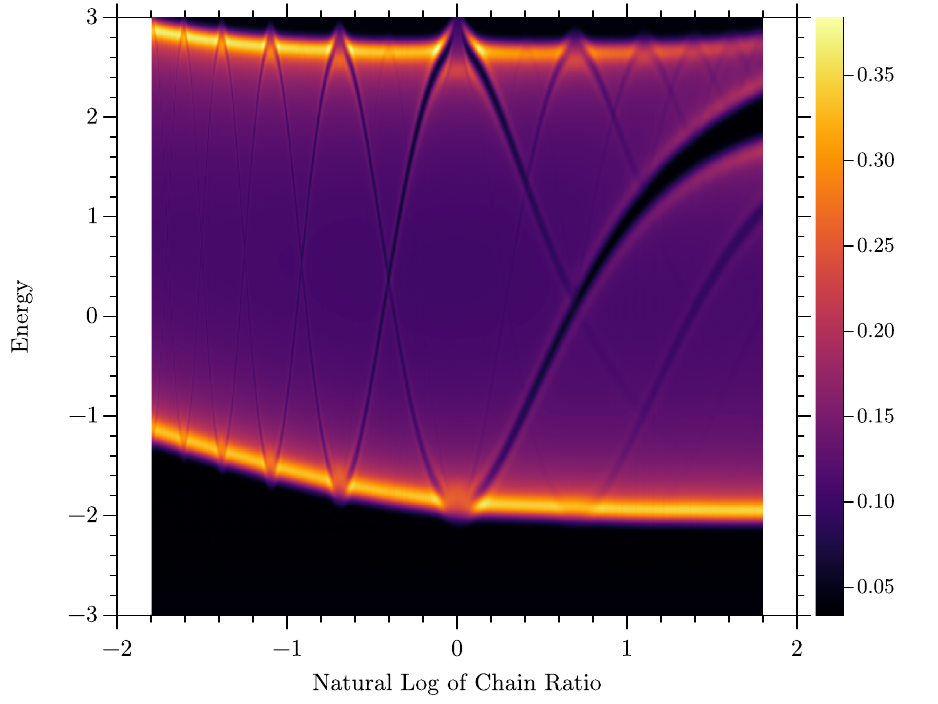}
    \caption{Density of states of the reduced incommensurate coupled chain operator \eqref{eq:reduced_single_chain} in color against natural log of interatomic distance ratio and energy level.}
    \label{fig:rcnumerics_Lap}
\end{figure}
\begin{figure}
    \centering
    \includegraphics[width=0.75\linewidth]{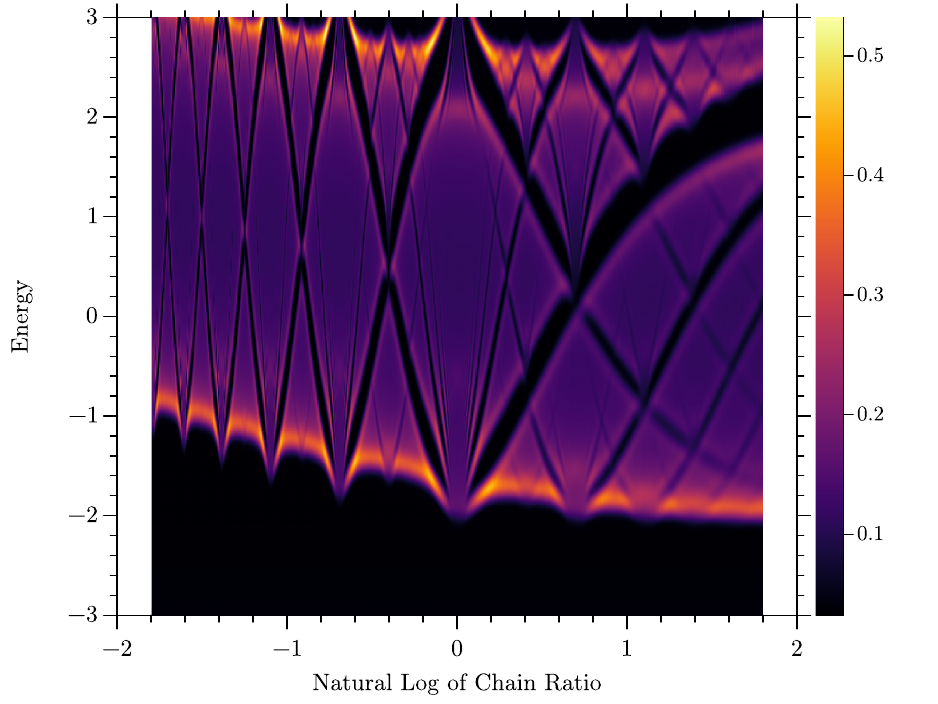}
    \caption{Density of states of the reduced incommensurate coupled chain operator with inverse Laplacian replaced by identity \eqref{eq:reduced_single_chain_noLap} in color against natural log of interatomic distance ratio and energy level.}
    \label{fig:rcnumerics}
\end{figure}
\begin{figure}
    \centering
    \includegraphics[width=0.75\linewidth]{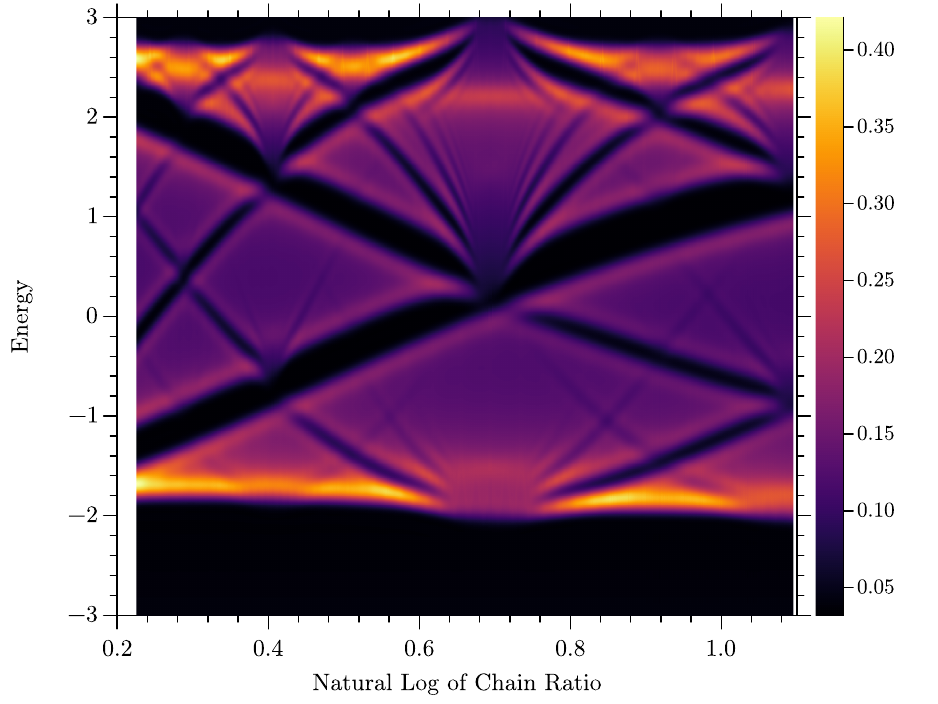}
    \caption{Higher resolution zoom of Figure \ref{fig:rcnumerics} showing chain ratios ranging from 1.25 to 3. Self similarity near the top and bottom of the spectrum suggests a fractal structure, especially for $\ln(p/q) \approx .7 \implies p/q \approx 2$. 
    \label{fig:rczoom}}
\end{figure}

\begin{figure}
    \centering
    \includegraphics[width=0.75\linewidth]{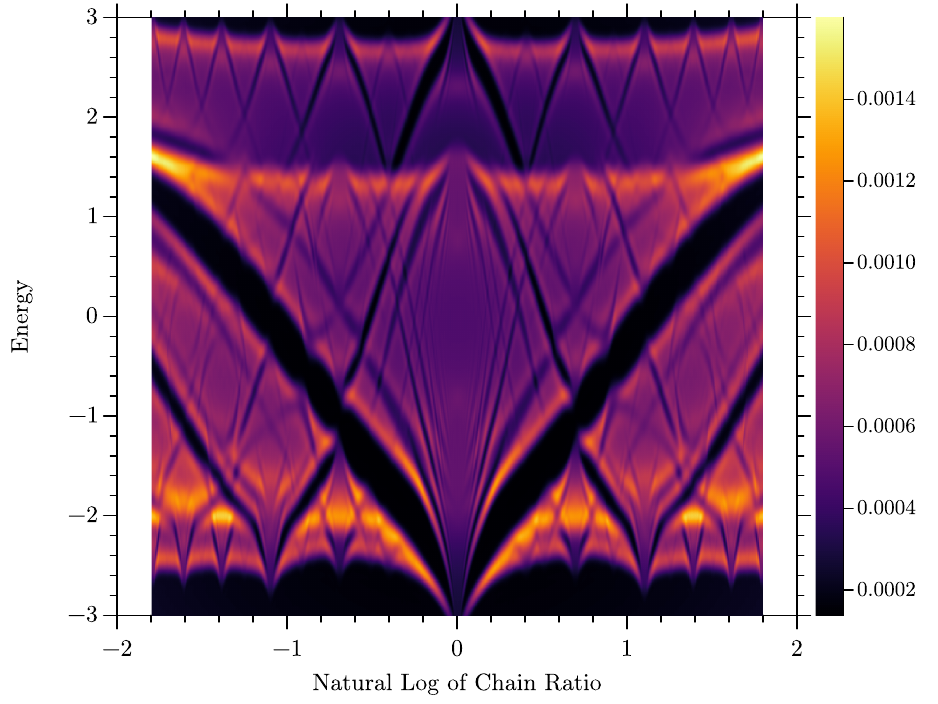}
    \caption{Density of states of the incommensurate coupled chain operator in color against natural log of interatomic distance ratio and energy level}
    \label{fig:ccnumerics}
\end{figure}

We show in Figure \ref{fig:ccnumerics} the computations of the density of states for the full incommensurate chain model \eqref{eq:coupled_chain_H} using the integral formula \eqref{eq:result_convergence_incommensurate_chain_operators}. To allow for comparison with the reduced and full coupled chain results, the parameters were kept the same. The results are consistent with calculations done by Canc\`es, Cazeaux, and Luskin in \cite{Cances2017a}, again showing a complex fractal-like structure reminiscent of the Hofstadter butterfly. In this case the structure is exactly symmetric due to the chains being interchangable. Note that in the calculations done in \cite{Cances2017a}, the hopping function was Gaussian. The similarity of our results shows that the precise form of the hopping function does not impact the density of states overly much.

We show in Figure \ref{fig:onsitenumerics} computations of the density of states for the incommensurate chain model with disorder \eqref{eq:rand_coup_chain} using the integral formula \eqref{eq:dosm_random_iccm}. The probability measure on $\mathbb{R}$ that we selected is the uniform measure on $[-1,1]$ for each site. We chose this so that the spectral radius of the onsite operator would not be much larger than the regular chain operator. Rather than carrying out the full integral over the probability distribution, we added random values from the distribution to each diagonal element of the matrix at each value of the disregistry. 
We expect that this approach converges to the true result as we increase the discretization of the disregistry.

The picture displays many of the same characteristics of the unmodified coupled chain, with large low density gaps in arcs extending from the bottom of the spectrum. However, much of the fine grained behavior near whole number ratio points like $\ln\left(2\right)$ (i.e., where $\ln(p/q) = \ln(2) \implies p/q = 2$) have mostly disappeared and been replaced with noise. The stronger gap at $\ln(3)$ is still faintly visible as well. The gaps still remaining in this image may represent more closely physically realizable gaps for materials at finite temperature.

\begin{figure}
    \centering
    \includegraphics[width=0.75\linewidth]{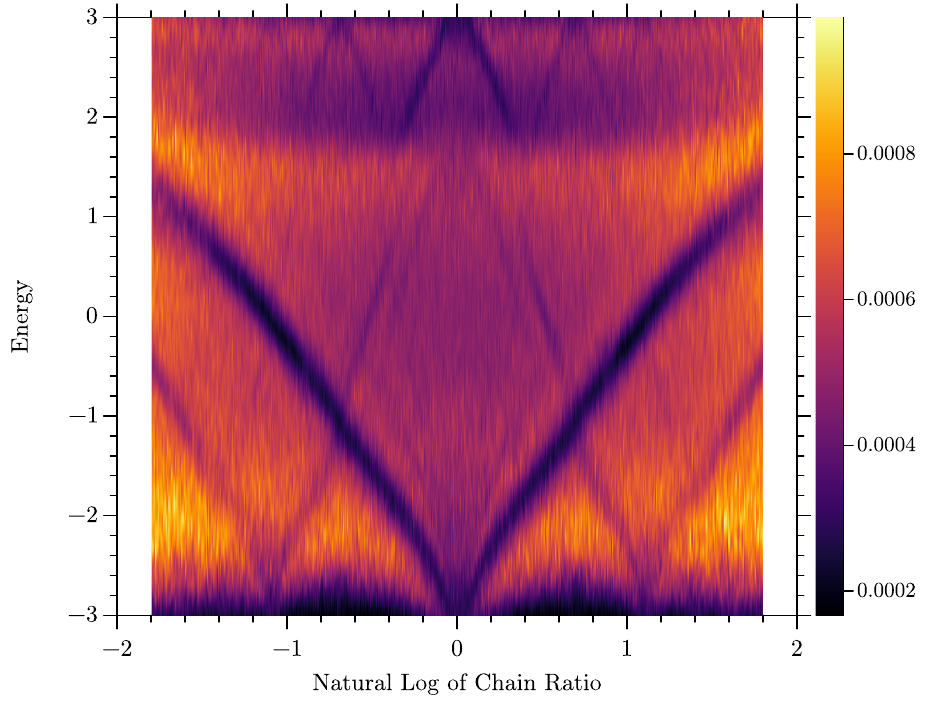}
    \caption{Density of states of the incommensurate coupled chain model with disorder against the interatomic distance ratio and energy level.}
    \label{fig:onsitenumerics}
\end{figure}
\bibliographystyle{plain}
\bibliography{references_math}

\appendix

\section{Proof of Lemma \ref{lemma:random_iccm_ergodic_shifting_incom_unit_cells}} \label{sec:proofs_random_iccm}

Let $(\mathbb{R},\mathcal{B}_0,\mathbb{P}_0)$ be a probability space where $\mathcal{B}_0$ is the Borel $\sigma$-algebra and $\mathbb{P}_0$ is compactly supported. We start by reviewing the proof of ergodicity for the simpler case of the product space $(\mathbb{R}^\mathbb{Z},\mathcal{B},\mathbb{P})$, where $\mathcal{B}$ and $\mathbb{P}$ are the product $\sigma$-algebra and measure, respectively (e.g. Exercise 3.2 of \cite{AizenmanWarzel2015}). Consider the shifts $(T_x)_{x \in \mathbb{Z}}$ acting on $\mathbb{R}^\mathbb{Z}$ by $\left( T_x \omega \right)_j = \omega_{j-x}$. We claim these transformations are \emph{mixing}, i.e.,
\begin{equation} \label{eq:mix}
    \mathbb{P}\left( A \cap T_x^{-1} B \right) = \mathbb{P}(A) \mathbb{P}(B) \text{ as $|x| \rightarrow \infty$}
\end{equation}
for all $A, B \in \mathcal{B}$. Note that ergodicity follows, since if $A \in \mathcal{B}$ is invariant, i.e., $T_x A = A$ for all $x \in \mathbb{Z}$, by evaluating \eqref{eq:mix} at $A = B = T_x^{-1} B$ we have
\begin{equation}
    \mathbb{P}(A) = \mathbb{P}(A)^2 \implies \mathbb{P}(A) \in \{0,1\}.
\end{equation}
To see that $(T_x)_{x \in \mathbb{Z}}$ are mixing, assume initially that $A$ and $B$ are \emph{finite-dimensional subsets}, i.e., equal to $\mathbb{R}$ for all but finitely-many components. Then, by taking $x$ sufficiently large, we can arrange that the components of $A$ and $T_x^{-1} B$ not equal to $\mathbb{R}$ are disjoint. We can then compute 
\begin{equation}
    \mathbb{P}( A \cap T_x B ) = \mathbb{P}( A ) \mathbb{P}( T_x B ) = \mathbb{P}( A ) \mathbb{P}( B ).
\end{equation}
To prove the result for arbitrary $A, B \in \mathcal{B}$, fix $\epsilon > 0$. Let $A_0, B_0$ be finite-dimensional so that $\mathbb{P}(A\Delta A_0) + \mathbb{P}(B\Delta B_0) < \epsilon$ (Section 13 Theorem D of \cite{halmos2013measure}). Then we can prove, writing $A = A_0 \cup ( A \Delta A_0 )$, $B = B_0 \cup ( B \Delta B_0 )$, that
\begin{equation}
    \left| \mathbb{P}( A \cap T_x B ) - \mathbb{P}( A_0 \cap T_x B_0 ) \right| < \epsilon. 
\end{equation}
By taking $|x|$ sufficiently large, we then have by the previous argument that 
\begin{equation}
    \left| \mathbb{P}( A \cap T_x B ) - \mathbb{P}( A_0 ) \mathbb{P}( B_0 ) \right| < \epsilon, 
\end{equation}
and then approximating again we have
\begin{equation}
    \left| \mathbb{P}( A \cap T_x B ) - \mathbb{P}( A ) \mathbb{P}( B ) \right| < 2 \epsilon,
\end{equation}
which suffices since $\epsilon > 0$ was arbitrary.

\begin{proof}[Proof of Lemma \ref{lemma:random_iccm_ergodic_shifting_incom_unit_cells}]
    We consider the transformations $\{\tilde{T}^1_{R,x}\}_{x\in \mathbb{Z}}$ without loss of generality. Recall that these transformations act on the product probability space $([0,1-\theta) \times \mathbb{R}^\mathbb{Z} \times \mathbb{R}^\mathbb{Z}, \mathcal{A}_2, \mathbb{P}_{R,2} )$ as
    \begin{equation}
        \tilde{T}_{R,x}^1(b,\omega_1,\omega_2) = \left(b+x - (1-\theta) [b+x]_{1-\theta}, (\omega_{1(j-x)})_{j \in \mathbb{Z}}, (\omega_{2(j - [b+x]_{1-\theta})})_{j \in \mathbb{Z}} \right).
    \end{equation}
    Here $\mathcal{A}_2, \mathbb{P}_{R,2}$ denote the product Borel $\sigma$-algebra and product measure between the probability spaces $([0,1-\theta),\mathcal{A},\mathbb{P}_2)$, where $\mathcal{A}$ is the Borel $\sigma$-algebra and $\mathbb{P}_2$ is the normalized Lebesgue measure, and $(\mathbb{R}^\mathbb{Z} \times \mathbb{R}^\mathbb{Z},\mathcal{B},\mathbb{P})$, where $\mathcal{B}$ is the product Borel $\sigma$-algebra and $\mathbb{P}$ is the product measure over the probability spaces $(\mathbb{R},\mathcal{B}_0,\mathbb{P}_0)$, where $\mathcal{B}_0$ is the Borel $\sigma$-algebra and $\mathbb{P}_0$ is compactly supported.
    That these transformations form a group is clear.

    We now essentially follow the proof that if a transformation $\sigma$ is mixing then $\sigma \times \tau$ is ergodic for every ergodic $\tau$ (Proposition 4.21 of \cite{Nadkarni1998}), but write the argument out in detail because of the $b$-dependence of the action of $\tilde{T}^1_{R,x}$ on $\mathbb{R}^\mathbb{Z} \times \mathbb{R}^\mathbb{Z}$. Ergodicity of the $\{\tilde{T}^1_{R,x}\}_{x\in \mathbb{Z}}$ is equivalent to the statement that
    \begin{equation} \label{eq:equiv_to_ergodic}
        \lim_{L \rightarrow \infty} \frac{1}{2L+1} \sum_{x = -L}^{L} \mathbb{P}_{R,2}\left( \tilde{T}^1_{R,x}(A) \cap B \right) = \mathbb{P}_{R,2}(A) \mathbb{P}_{R,2}(B) 
    \end{equation}
    for all $A, B \in \mathcal{A}_2$ \cite{Nadkarni1998}. Since the product $\sigma$-algebra is generated by the algebra of finite disjoint unions of rectangles (Section 33 Theorem E of \cite{halmos2013measure}), by an approximation argument as above (Section 13 Theorem D of \cite{halmos2013measure}), it suffices to prove \eqref{eq:equiv_to_ergodic} for all rectangles $A = \alpha \times \beta$, $B = \gamma \times \delta$, for $\alpha, \gamma \in \mathcal{A}$ and $\beta, \delta \in \mathcal{B}$. By further approximation we can assume that $\beta, \delta \in \mathcal{B}$ are finite-dimensional subsets in the sense of the previous section, i.e., all but finitely-many components are equal to $\mathbb{R}$. The $\tilde{T}^1_{R,x}$ act on such sets by
    \begin{equation}
        \tilde{T}^1_{R,x}(\alpha \times \beta) = \tilde{T}^1_{R,x,1}(\alpha) \times \tilde{T}^1_{R,x,2}(\beta;\alpha),
    \end{equation}
    where $\tilde{T}^1_{R,x,1}(b) = b + x - (1-\theta)[b+x]_{1-\theta}$ and $\tilde{T}^1_{R,x,2}(\omega_1,\omega_2;b) = \left( \left(\omega_{1(j-x)}\right)_{j \in \mathbb{Z}} , \left(\omega_{2(j-[b+x]_{1-\theta})}\right)_{j \in \mathbb{Z}} \right)$. By definition, we then have 
    \begin{equation}
        \begin{split} 
            \frac{1}{2L+1} \sum_{x = -L}^{L} \mathbb{P}_{R,2}\left( \tilde{T}^1_{R,x}(A) \cap B \right) &= \frac{1}{2L+1} \sum_{x = -L}^{L} \mathbb{P}_{R,2}\left( \tilde{T}^1_{R,x}(\alpha \times \beta) \cap (\gamma \times \delta) \right)   \\
            &= \frac{1}{2L+1} \sum_{x = -L}^{L} \mathbb{P}_{R,2}\left( \left( \tilde{T}^1_{R,x,1}(\alpha) \cap \gamma \right) \times \left( \tilde{T}^1_{R,x,2}(\beta;\alpha) \cap \delta \right) \right).
        \end{split}
    \end{equation}
By definition of the product measure we have
    \begin{equation}
        = \frac{1}{2L+1} \sum_{x = -L}^{L} \mathbb{P}_{2}\left( \tilde{T}^1_{R,x,1}(\alpha) \cap \gamma \right) \cdot \mathbb{P}\left( \tilde{T}^1_{R,x,2}(\beta;\alpha) \cap \delta \right),
    \end{equation}
which we write as
    \begin{equation}
        = \frac{1}{2L+1} \sum_{x = -L}^{L} \mathbb{P}_{2}\left( \tilde{T}^1_{R,x,1}(\alpha) \cap \gamma \right) \cdot \left( \mathbb{P}\left( \tilde{T}^1_{R,x,2}(\beta;\alpha) \cap \delta \right) - \mathbb{P}(\beta) \mathbb{P}(\delta) \right) + \mathbb{P}_{2}\left( \tilde{T}^1_{R,x,1}(\alpha) \cap \gamma \right) \cdot \mathbb{P}(\beta) \mathbb{P}(\delta).
    \end{equation}
    But now notice that, just as in the previous section, since $\left| [b+x]_{1-\theta} \right|$ can be bounded below by $|x| - |1 - \theta|$, we have
    \begin{equation}
        \mathbb{P}\left( \tilde{T}^1_{R,x,2}(\beta;\alpha) \cap \delta \right) = \mathbb{P}(\beta) \mathbb{P}(\delta) \text{ for $|x|$ sufficiently large,}
    \end{equation}
and hence the first term can be made arbitrarily small for large $L$. But now notice that the second term satisfies
    \begin{equation}
        \lim_{L \rightarrow \infty} \frac{1}{2L+1} \sum_{x = -L}^{L} \mathbb{P}_{2}\left( \tilde{T}^1_{R,x,1}(\alpha) \cap \gamma \right) \cdot \mathbb{P}(\beta) \mathbb{P}(\delta) = \mathbb{P}_2(\alpha)\mathbb{P}_2(\gamma)\mathbb{P}(\beta)\mathbb{P}(\delta) = \mathbb{P}_{R,2}(A) \mathbb{P}_{R,2}(B)
    \end{equation}
    where the first equality follows immediately from ergodicity of the transformations $\tilde{T}^1_{R,x,1}$ and the second from the definition of the product measure.

\end{proof}

\section{Kernel polynomial method} \label{sec:KPM}

To visualize the density of states across the spectrum of $H$, we evaluate the density of states measure with $g$ given by convenient smoothings of the delta distribution shifted along a grid of energies $E$. Recall that the Chebyshev polynomials are defined by the recurrence relation
\begin{align*}
    C_{i+2}(x) = 2x&C_{i+1}(x) -C_i(x) \\
    C_0 = 1 ~&~ C_1 = x
\end{align*}
and satisfy the orthogonality condition
\begin{align*}
    \int_{-1}^1 \frac{1}{\pi \sqrt{1-x^2}}T_n(x)T_m(x)dx = \frac{1+\delta_{0n}}{2}\delta_{nm}.
\end{align*}
Using these polynomials, we can introduce the functions
\begin{equation}
    G(x,E;N) := \frac{1}{\pi \sqrt{1-E^2}}\sum_{i\leq N} C_i(x) C_i(E)~~~~~~ (x,E)\in (-1,1)
\end{equation}
which approximate the delta distribution centered at $E$ for $(x,E) \in (-1,1)$ where $N$ is the highest degree Chebyshev polynomial used in the series. To approximate the delta distribution on the interval $\left(-\frac{1}{\eta},\frac{1}{\eta}\right)$, we introduce the scaled functions
\begin{align*}
    \tilde{G}(x,E;\eta,N) := \eta G(\eta x,\eta E;\eta,N) = \frac{\eta}{\pi \sqrt{1-(\eta E)^2}}\sum_{i\leq N} C_i(\eta x)C_i(\eta E) ~~~~~~ (x,E)\in \left(-\frac{1}{\eta},\frac{1}{\eta}\right).
\end{align*}

We obtain a na\"ive numerical approximation of the density of states by evaluating \eqref{eq:red_this}, \eqref{eq:pre_result_convergence_incommensurate_chain}, and \eqref{eq:pre-result_convergence_dos_random_iccm} with the function $g$ given by $\tilde{G}$, with a suitable truncation of $H(b)$ substituted for $x$, calculating this scalar product over a range of $E$ values, and with $\eta$ chosen so that the interval $\left(-\frac{1}{\eta},\frac{1}{\eta}\right)$ always contains the spectrum.
However, this approximation is very susceptible to Gibbs oscillations. These oscillations heavily impact the calculation of the density of states by making the functions $\tilde{G}$ negative on some intervals. This can be fixed by introducing the Jackson coefficients
\begin{equation*} \label{JacksonCoeffs}
    J_i^N := \hspace{.2cm}\frac{(N-i+1)\cos\left(\frac{\pi i}{N+1}\right)+\sin\left(\frac{\pi i}{N+1}\right) \cot\left(\frac{\pi}{N+1}\right)}{N+1},
\end{equation*}
and replacing the functions $\tilde{G}$ by
\begin{align*}
    \tilde{G}_J(x,E;\eta,N) := \frac{\eta}{\pi \sqrt{1-(\eta E)^2}}\sum_{i\leq N} J^N_i C_i(\eta x)C_i(\eta E).
\end{align*}
These coefficients transform the Chebyshev expansion to make it positive everywhere and remove the oscillatory behavior. See Figure \ref{fig:jackson} for a plot showing the difference that the coefficients make. 

The computation of density of states is made more efficient by observing that
\begin{align*}
    \langle \delta_{x}, \tilde{G}_J (E,H;\eta)\delta_{x}\rangle = \frac{\eta}{\pi \sqrt{1-(\eta E)^2}}\sum_{i\leq N} J^N_i C_i(\eta E)\langle \delta_{x},C_i(\eta H)\delta_{x}\rangle.
\end{align*}
This means that the diagonal matrix elements of the Hamiltonian can be re-used when computing the density of states at different values of $E$. This leads to large savings in numerical calculations.

We return to the Chebyshev recursion formula to simplify calculating the large degree matrix polynomials we need as
\begin{align*}
    C_{i+2}(x) = 2x&C_{i+1}(x) -C_i(x) \implies C_{i+2}(H) = 2HC_{i+1}(H) -C_i(H) \\
    C_0 = 1 ~&~ C_1 = x ~~~~~~~~~~~~~~~~~~~~~~~~~~C_0 = I ~~ C_1 = H.
\end{align*}
We can use this to generate arbitrarily high Chebyshev polynomials of our Hamiltonians quickly. Since we only need a single element from the Hamiltonian, much of the calculation can be simplified further by manipulating the seeds of the polynomial recursion as follows
\begin{align*}
    D_{i+2}(H) = 2H&D_{i+1}(H) -D_i(H) \\
    D_0 = v_0 ~&~ D_1 = Hv_0
\end{align*}
where $v_0$ is a vector of zeros that has a 1 in the location of the desired element. This greatly reduces the complexity once we make a truncation
\begin{align*}
    \langle \delta_{x}, C_i(\eta H)\delta_{x}\rangle = v_0 \cdot D_i(\eta H).
\end{align*}

\begin{figure}
    \centering
    \includegraphics[width = 3.15in]{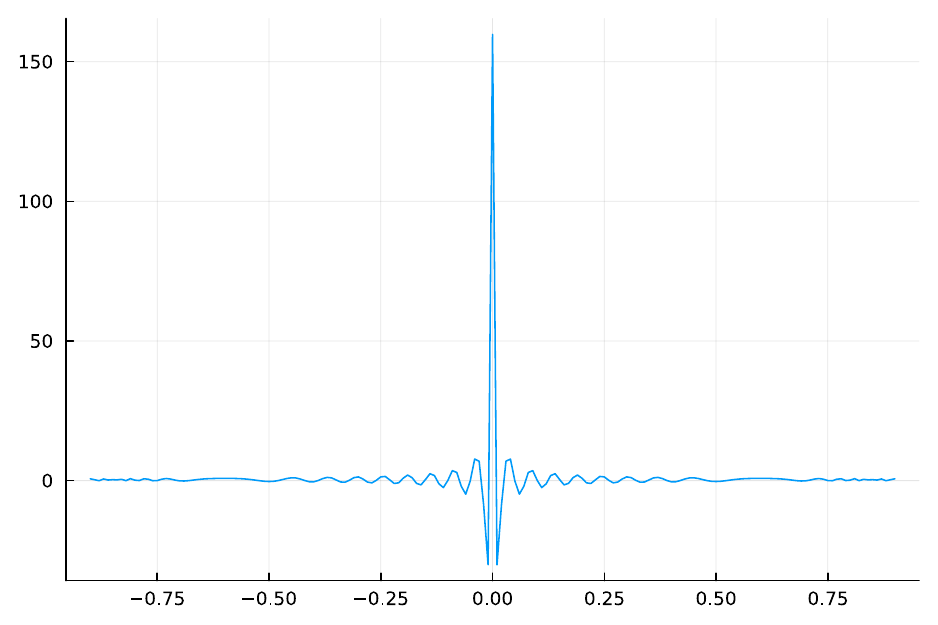}
    \includegraphics[width = 3.15in]{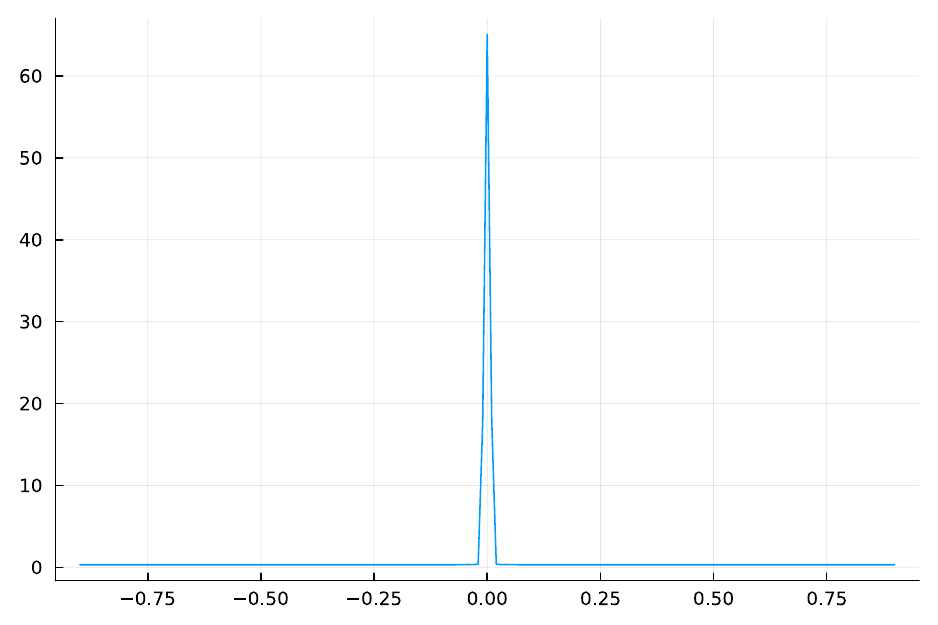}
    \caption{Approximations of $\delta(x)$ up to order 500 Chebyshev polynomials. Left: without Jackson Coefficients, Right: with Jackson Coefficients}
    \label{fig:jackson}
\end{figure}

\end{document}